\documentclass[sigconf]{acmart}
\usepackage{latexsym,amsfonts}

\usepackage{tikz}
\usetikzlibrary{calc,shapes,arrows,automata}
\usepackage{amsmath}
\usepackage{algorithm}
\usepackage{algorithmic}
\usepackage{url}
\usepackage{subcaption} 
\usepackage{caption}
\newcommand{\R}{\mathbb{R}}
\newcommand{\N}{\mathbb{N}}

\newcommand{\Bb}{\mathcal{B}}

\newcommand{\Xx}{\mathcal{X}}

\newcommand{\Sys}{\mathfrak{S}}
\newcommand{\Ll}{\mathcal{L}}

\newcommand{\Inf}{\mathsf{Inf}}
\newcommand{\Tt}{\mathcal{T}}

\newcommand{\Aa}{\mathcal{A}}

\newcommand{\set}[1]{\{ #1 \}}
\setcopyright{rightsretained}

\newtheorem{definition}{Definition}[section]
\newtheorem{theorem}{Theorem}

\copyrightyear{2024}
\acmYear{2024}
\setcopyright{rightsretained}
\acmConference[HSCC '24]{27th ACM International Conference on Hybrid
Systems: Computation and Control}{May 14--16, 2024}{Hong Kong, Hong Kong}
\acmBooktitle{27th ACM International Conference on Hybrid Systems:
Computation and Control (HSCC '24), May 14--16, 2024, Hong Kong, Hong
Kong}\acmDOI{10.1145/3641513.3650120}
\acmISBN{979-8-4007-0522-9/24/05}
\begin{document}

\title{Closure Certificates}
\author{Vishnu Murali}
\affiliation{%
	\institution{University of Colorado Boulder}
	\city{Boulder}
	\country{USA}
}
\email{vishnu.murali@colorado.edu}

\author[]{Ashutosh Trivedi}
\affiliation{%
	\institution{University of Colorado Boulder}
	\city{Boulder}
	\country{USA}
}
\email{ashutosh.trivedi@colorado.edu}

\author[]{Majid Zamani}
\affiliation{%
	\institution{University of Colorado Boulder}
	\city{Boulder}
	\country{USA}
}
\email{majid.zamani@colorado.edu}

\begin{abstract}
    \label{sec:abstract}
    A \emph{barrier certificate}, defined over the states of a dynamical system, is a real-valued function whose zero level set characterizes an inductively verifiable {\it state invariant} separating reachable states from unsafe ones.  
    When combined with powerful decision procedures---such as sum-of-squares programming (SOS) or satisfiability-modulo-theory solvers (SMT)---barrier certificates enable an automated deductive verification approach to safety.
    The barrier certificate approach has been extended to refute LTL and $\omega$-regular specifications by separating consecutive transitions of corresponding $\omega$-automata in the hope of denying all accepting runs.
    Unsurprisingly, such tactics are bound to be conservative as refutation of \emph{recurrence properties} requires reasoning about the well-foundedness of the transitive closure of the transition relation.
    This paper introduces the notion of \emph{closure certificates} as a natural extension of barrier certificates from state invariants to transition invariants. 
    We augment these definitions with SOS and SMT based characterization for automating the search of closure certificates and demonstrate their effectiveness over some case studies.
\end{abstract}
\maketitle

\begin{CCSXML}
\end{CCSXML}

\section{Introduction} 
\label{sec:intro}
As cyber-physical systems and internet-of-things continue to proliferate within critical infrastructure, the need for practical verification algorithms for infinite-state dynamical systems is ever-present. 
Structural induction over the transition structure of dynamical systems provides a lightweight yet powerful proof method to establish safety and invariance guarantees. 
However, when the invariant is not inductive, human ingenuity is required in strengthening the invariant to an inductive one. 
The notion of {\it barrier certificates}~\cite{prajna_2004_safety}, when combined with automatic decision procedures automate the search for an inductive state invariant. 
This paper presents \emph{closure certificates} as a generalization of barrier certificates to capture the \emph{transitive closure} of transition relations to automate verification of linear temporal logic (LTL) and $\omega$-regular specifications of discrete-time dynamical systems.

\vspace{0.5em}\noindent\textbf{Barrier Certificates for State Invariants.} Intuitively, a barrier certificate~\cite{prajna_2004_safety} is a real-valued function over the state space that is negative over the initial states, positive over the unsafe states, and it does not increase with transitions. 
From this definition and the principle of structural induction, it follows that the zero level set of the barrier certificate over-approximate the set of reachable states. 
This, together with the positivity requirement over the unsafe states, provide a separation between reachable and unsafe states, guaranteeing safety. 
The results in~\cite{wongpiromsarn_2015_automata} extended the barrier certificate based approach to refute violations of linear temporal logic (LTL) specifications expressed via $\omega$-automata. 
In this so-called \emph{state-triplet} approach, barrier certificates provide separation between consecutive transitions (involving three states) of the given $\omega$-automaton in such a way that denies accepting runs.
The approach has been extended for verification and synthesis for more general dynamical systems~\cite{jagtap_2018_temporal,jagtap_2020_formal,anand_2021_compositional,anand_2022_small}.
These state-triplet approach are bound to suffer from conservatism as the verification of a general $\omega$-regular property requires refutation of infinitely many visits to some state and that in turn requires a well-founded argument~\cite{podelski_2004_transition,cook_2009_priciples} over transitive closure of transition relation.

\vspace{0.5em}\noindent\textbf{Closure Certificates for Transition Invariants.} Podelski and Rybalchenko, in an influential paper~\cite{podelski_2004_transition}, introduced  disjunctively well-founded \emph{transition invariants} to verify programs against $\omega$-regular properties. 
They defined the transition invariant as an over-approximation of the transitive closure of the transition relation, restricted to the set of reachable states. 
If the transition invariant restricted to pairs of accepting states is disjunctively well-founded, then they showed that no execution can visit these accepting states infinitely often, refuting the $\omega$-regular specification. 
We introduce closure certificates as a functional analog of \emph{transition invariants} and enable the use of SOS programming and SMT solvers to search for these certificates.

Intuitively, a closure certificate $\Tt: \Xx \times \Xx \to \R$ is a real-valued function over the Cartesian product of the state set and itself (state pairs), such that $\Tt(x, x') \geq 0$ if $x'$ is reachable from $x$. 
The closure certificate characterizes a transition invariant $T \subseteq \Xx \times \Xx$, with the set of initial states $\Xx_0$, in the following fashion: 
\begin{equation}
    T = \set{ (x, x') :  \Tt(x, x') {\geq} 0 \text{ and } \Tt(x_0, x) {\geq} 0 \text{ for some $x_0 \in \Xx_0$}}. \label{eq11}
\end{equation}

It is easy to see (Theorem~\ref{thm:bar_to_close}) that the existence of a barrier certificate implies the existence of a closure certificate establishing the same property. 
On the other hand, to appreciate the utility of closure certificate, we show that, even for safety properties (state inviariants), it is often possible to construct a closure certificate of simpler shape (e.g., lower degree polynomials) than a barrier certificate. 
To demonstrate this, consider the simple finite state discrete example shown in Figure~\ref{fig:sys_eg}. 
Here we depict initial states with green filled circles ($\Xx_0 = \set{1,3, 5}$) while unsafe states are shown with red filled circles ($\Xx_u = \set{2, 4}$).
It is easy to see that starting from the initial states, the system never visits any unsafe states. 
We show that while there is no \emph{polynomial} barrier certificate of degree $2$ that demonstrates the safety of the system, there is a \emph{linear} closure certificate that does so. 
We note that this example can be modified to show the absence of barrier certificate for any fixed degree.

\begin{figure}[t!]
  \begin{center}
  \begin{tikzpicture}[node distance =2cm]
    \node[state, draw, initial text =,fill=blue!10!white] (0) at (0,0) {$0$};
    \node[state, fill=green!10!white,] (1) at (-1,-1.2) {$1$};
    \node[state, fill=red!20!white,] (2) at (0,-1.2) {$2$};
    \node[state, fill = green!10!white] (3) at (1,-1.2) {$3$};
    \node[state, fill = red!20!white] (4) at (2,-1.2) {$4$};
    \node[state, fill = green!10!white] (5) at (3,-1.2) {$5$};
    \node[] (6) at (-1, -2.2) {};
    \node[] (7) at (1, -2.2) {};
    \node[] (8) at (3, -2.2) {};
    \path[->]
    (1) edge node[above]{} (0)
    (2) edge node{} (0)
    (3) edge node{} (0)
    (4) edge node{} (0)
    (5) edge[bend right] node{} (0)
    (0) edge[loop above] node{} (0)
    (6) edge node{} (1)
    (7) edge node{} (3)
    (8) edge node{} (5);
    \end{tikzpicture}
  \end{center}
  \caption{Illustrative example demonstrating the simplicity of closure certificates over barrier certificates}
  \label{fig:sys_eg}
\end{figure}
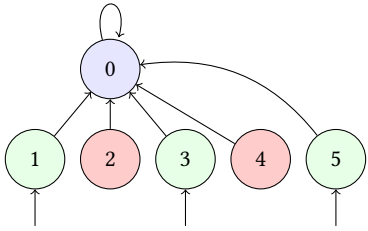

Let us suppose that there exists a polynomial barrier certificate $\Bb$ of degree $2$ that acts as a proof of safety.
We need $\Bb(x) \leq 0$ for every state $x \in \Xx_0$, and $\Bb(x) > 0$ for every state $x \in \Xx_u$.
Applying intermediate value theorem, the function $\Bb$ needs to change signs in at least $3$ points and must therefore have at least $3$ roots. 
This supports our claim that there is no barrier certificate that is a polynomial of degree $2$.
On the other hand, the linear function $\Tt: \Xx \times \Xx \to \R$ defined as $\Tt(x,y) = -y$ is a closure certificate for this system.
While we defer the details to later sections, from~(\ref{eq11}), it follows that this function corresponds to the transition invariant $T = \set{(1,0), (3,0), (5,0), (0,0) }$, and has no intersection with the set $\Xx_0 \times \Xx_u$.
This ensures the safety of the system.
We have chosen a finite-state example for illustrative purposes. This paper deals with continuous state spaces, and our case studies (cf. Section~\ref{sec:case_studies}) will demonstrate similar advantages in continuous state spaces.

While barrier certificates can be employed to verify other, more complex objectives (e.g., liveness or general linear-time properties), their applications in such settings are often conservative~\cite{jagtap_2018_temporal, jagtap_2020_formal, anand_2021_compositional}. We adapt closure certificates to validate or refute general linear-time properties. As an example, consider the so-called ``persistence'' property, where one wishes to verify that a system visits some region (denoted as $\Xx_{VF}$) only finitely often, or alternatively, it eventually stays within some region (the complement of $\Xx_{VF}$). 
We extend closure certificate conditions (Section\ref{subsec:finite_vis}) with a ``potential''-like argument. In particular, we require that for every initial state $x_0$ and every pair of states $y$ and $y'$ in the set $\Xx_{VF}$, if the system can reach from $y$ to $y'$, then the potential between $(x_0, y')$ is less than the potential between $(x_0, y)$ by a certain fixed amount. This, in turn, implies that for every execution starting from an initial state, the region $\Xx_{VF}$ can only be visited finitely often.
This approach can be extended to general linear-time objectives (Section~\ref{subsec:ltl_verif}) by employing the classical automata-theoretic approach that reduces LTL verification to visiting certain states only finitely often.

\vspace{0.5em}
\noindent \textbf{Contributions.}
The contributions of the paper are listed next.
\begin{enumerate}
    \item  This paper proposes a novel notion of closure certificates that act as a functional analog to transition invariants.
    \item We present SOS programming as well as SMT characterizations to search for a closure certificates within a given template (function class).
    \item We show that even when traditional barrier certificates of a some template fail to ensure safety, one can find closure certificates of the same template.
    \item We demonstrate how to use closure certificates to verify dynamical systems against LTL specifications described by $\omega$-automata with our case studies.
    \item We show how closure certificates subsume existing barrier certificate based approaches to verify continuous-space dynamical systems against LTL specifications.
\end{enumerate}

\vspace{0.5em} \noindent \textbf{Related works.}
Prajna and Jadbabaie proposed the notion of barrier certificates~\cite{prajna_2004_safety} as a discretization-free approach to give guarantees of safety or reachability~\cite{prajna_2007_convex} for dynamical and hybrid systems.
The results in~\cite{wongpiromsarn_2015_automata} presented a state triplet approach that uses barrier certificates to verify linear temporal logic properties specified by $\omega$-automata. 
This approach has since been used in the the verification and synthesis of stochastic and interconnected continuous-space  systems against linear temporal logic properties~\cite{jagtap_2018_temporal,jagtap_2020_formal,anand_2021_compositional,anand_2022_small}. 
Unfortunately, the above approach is conservative in the sense that it treats the  nondeterministic B\"uchi automaton corresponding to the negation of the LTL specification as a finite automaton and then searches for barrier certificates to disallow the transitions along accepting paths to show the accepting state is not visited.
Thus, even if a system satisfies the property but visits the accepting state, then one cannot make use of the above approaches to verify the system.
We show (cf. Section~\ref{sec:subsumption}) that our approach subsumes this current approach.
Podelski and Rybalchenko~\cite{podelski_2004_transition}, proposed a notion of transition invariants and demonstrated their use in verifying the liveness properties of programs as well as in verifying programs against $\omega$-regular properties.
Transition invariants have also been used in~\cite{podelski_2006_model}  to give guarantees of stability for hybrid systems. 
Here, they make use of a reachability analysis tool to determine the overapproximation of reachable states and then establish a Lyapunov guarantee on the transition invariant to ensure the stability of the system. The results in~\cite{sankaranarayanan_2011_relational} consider 
a notion of relational abstraction that is similar to transition invariants to give guarantees for safety. 

While this paper focuses on abstraction-free approaches to verify LTL properties specified as $\omega$-automata, we note that abstraction-based techniques have been used in the verification and synthesis of continuous-space dynamical systems against LTL properties such as the results in ~\cite{henzinger_1997_hytech,lahijanian_2011_temporal,rungger_2016_scots,khaled_2021_omegathreads}.  These results rely on building a finite state abstraction and then making use of model checking and synthesis techniques~\cite{tabuada_2009_verification,baier_2008_principles}.

\section{Preliminaries}
\label{sec:prelims}
We use $\N$ and $\R$ to denote the set of natural numbers and reals.
For $a \in \R$, we use $\R_{\geq a}$ and $\R_{> a}$ to denote the intervals $[a, \infty)$ and $(a,\infty)$, respectively, and similarly, for any natural number $n \in \N$, we use $\N_{\geq n}$ to denote the set of natural numbers greater than or equal to $n$.
Given a set $A$, sets $A^{*}$ and $A^{\omega}$ denote the set of finite and countably infinite sequences of elements in $A$, while $|A|$ denotes the cardinality of the set. 
If $A \subseteq B$, and the set $B$ can be inferred from the context, we denote the complement $B \setminus A$ simply as $\overline{A}$. 
We call a function $f:A \to \R$ {\it bounded} if there exists $l,u \in \R$, such that $ l \leq f(a) \leq u$ for every $a \in A$.

Given a relation $R \subseteq A \times B$, and an element $a \in A$, we use $R(a)$ to denote the set $\{ b \mid (a,b) \in R \}$. 
We use the notation $(a_1, a_2, \ldots, a_n) \in A^{*}$ for finite length sequences and $\langle a_0, a_1, \ldots, \rangle\in A^{\omega}$ for $\omega$-sequences.
Let $\Inf(s)$ be the set of infinitely often occurring elements in the sequence $s  =\langle a_0, a_1, \ldots, \rangle$.
Given an infinite sequence $s = \langle a_0, a_1, \ldots, \rangle$, and two natural numbers $i,j \in \N$ where $i \leq j$, we use $s[i,j]$ to indicate the finite sequence $(a_i, a_{i+1}, \ldots, a_j)$, and $s[i, \infty)$ to indicate the infinite sequence $\langle a_i, a_{i+1}, \ldots \rangle$. 
Finally, we use $s[i]$ to denote the $i$th element in the sequence $s$, \textit{i.e.}, given an infinite sequence $s = \langle a_0, a_1, \ldots \rangle$, we have $s[i] = a_i$ for any $i \in \N$. 
\subsection{Discrete-time Dynamical System}
\label{subsec:prelims_system}
A (discrete-time dynamical) system $\Sys$ is a tuple $(\Xx,\Xx_0, f)$, where $\Xx $ (possibly infinite) denotes the state set, $\Xx_0 \subseteq X$ denotes a set of initial states, and $f \subseteq \Xx \times \Xx$ is the state transition relation.
The state evolution of the system is given as the following:
\begin{equation}
\label{eq:state_evolution}
    \Sys: x(t+1) \in f(x(t)).
\end{equation}
If for every $x \in \Xx$, we have $|f(x)| = 1$, then we consider the transition relation $f$ to be a \emph{state transition function} that uniquely determines the next state.
Abusing notation, we use $f$ for both a set-valued map when it is a relation, and a transition function when it is a function.
Throughout the paper, we assume that state sets of the systems under consideration are compact.

A \emph{state sequence} is an infinite sequence  $\langle x_0, x_1, \ldots, \rangle \in \Xx^{\omega}$ where $x_0 \in \Xx_0$, and $x_{i+1} \in f(x_i)$, for all $i \in \N$.
We associate a labelling function $\Ll: \Xx \to \Sigma$ which maps each state of the system to a letter in a finite alphabet $\Sigma$. 
This naturally generalizes to mapping a state sequence of the system $\langle x_0, x_1, \ldots, \rangle \in \Xx^{\omega}$ to a trace or word $w = \langle \Ll(x_0), \Ll(x_1), \ldots, \rangle \in \Sigma^{\omega}$. 
For notational convenience, given a state $x \in \Xx$, we use $x'$ to indicate a state in $f(x)$. 
Let $TR(\Sys, \Ll)$ denote the set of all traces of $\Sys$ under the labeling map $\Ll$. 

\subsection{Specifications}
We are interested in deductive verification of linear-time properties over discrete-time dynamical systems.
We study increasingly complex specifications from safety, and persistence, to LTL and $\omega$-regular specifications.

\vspace{0.5em}\noindent \textbf{Safety.}
We say that a system is safe with respect to a set of unsafe states $\Xx_u \subseteq \Xx$ if for any state sequence $\langle x_0, x_1, \ldots, \rangle$ we have $x_i \notin \Xx_u$ for all $i \in \N$.
An important technique to verify the safety of the system is to synthesize \emph{barrier certificates}~\cite{prajna_2004_safety}.
\begin{definition}[Barrier Certificate]
\label{def:bar}
   A function $\Bb: \Xx \to \R$ is a barrier certificate for $\Sys = (\Xx, \Xx_0, f)$ with respect to a set of unsafe states $\Xx_u$ if: 
    \begin{align}
        & \Bb(x) \leq 0 && \text{ for all } x \in \Xx_0 \label{eq:bar_cond_1} \\
        & \Bb(x) > 0 && \text{ for all } x \in \Xx_u \label{eq:bar_cond_2} \\
        & \big( \Bb(x) \leq 0 \big) {\implies} \big( \Bb(x') \leq 0 \big) && \text{ for all } x \in \Xx, \text{ and } x' \in f(x) \label{eq:bar_cond_3}
    \end{align}
\end{definition}
\begin{theorem}[Barrier certificates imply safety \cite{prajna_2004_safety}]
    For a system $\Sys$ with unsafe states $\Xx_u$, the existence of a barrier certificate $\Bb$ implies its safety.
\end{theorem}

\vspace{0.5em}\noindent \textbf{Persistence  (refuting recurrence).}
We say that a system visits a region $\Xx_{VF} \subseteq \Xx$ only finitely often if for any state sequence $\langle x_0, x_1, \ldots, \rangle$ there exists some $i \in \N$, such that for all $j \geq i$, $j \in \N$, we have $x_j \notin \Xx_{VF}$.
Observe that if a system is safe with respect to a set of unsafe states $\Xx_u$, then it satisfies the persistence objective.
Thus one can make use of barrier certificates as a sound (not complete) way to ensure persistence.
Another approach to ensure persistence is to fix the value of $i$ to some constant value, and then search for a barrier certificate over the system and an augmented value.
Such approaches are common in bounded verification and synthesis approaches as in~\cite{schewe_2007_bounded, filiot_2009_antichain} for finite state systems.

\vspace{0.5em}\noindent \textbf{Linear Temporal Logic (LTL).}
Formulae in LTL~\cite{pnueli_1977_temporal}  are defined with respect to a set of finite atomic propositions $AP$ that are relevant to our system. 
Let $\Sigma = 2^{AP}$ denote the powerset of atomic propositions.
A trace $w = \langle w_0, w_1, \ldots, \rangle \in \Sigma^{\omega}$ is an infinite sequence of sets of atomic propositions.
The syntax of LTL can be given via the following grammar: 
\[\phi := \top \;|\; a \;|\; \neg \phi \;|\; \mathsf{X} \phi  \;|\; \phi \mathsf{U} \phi , \]
where $\top$ indicates true, $a \in AP$ denotes an atomic proposition, symbols $\wedge$, $\neg$ denote the logical AND and NOT operators respectively.
The temporal operators next, and until are denoted by $\mathsf{X}$, and $\mathsf{U}$ respectively.
The above operators are sufficient to derive the logical OR ($\vee$) and implication ($\implies$), and the temporal operators release ($\mathsf{R}$), eventually ($\mathsf{F}$) and always ($\mathsf{G}$) respectively.

We inductively define the semantics of an LTL formula with respect to trace $w$ as follows:
\begin{align}
&w \models a && \text{ if } a \in w[0] \\
& w \models \phi_1 \wedge \phi_2 && \text{ if } w \models \phi_1 \text{ and } w \models \phi_2 \\
& w \models \neg \phi && \text{ if } w \not\models \phi \\
& w \models \mathsf{X} \phi & & \text{ if } w[1, \infty) \models \phi\\
& w \models \phi_1 \mathsf{U} \phi_2 && \text{ if there exists } i \in \N \text{ such that } w[0,i] \models \phi_1 \nonumber \\ & && \text{ and } w[i+1, \infty) \models \phi_2
\end{align}
To reason about whether a system satisfies a property specified in LTL, we associate a labelling function $\Ll: \Xx \to \Sigma$ which maps each state of the system to a letter in the finite alphabet $\Sigma$. 
This naturally generalizes to mapping a state sequence of the system $\langle x_0, x_1, \ldots, \rangle \in \Xx^{\omega}$ to a trace $w = \langle \Ll(x_0), \Ll(x_1), \ldots, \rangle \in \Sigma^{\omega}$. 
Let $TR(\Sys, \Ll)$ denote the set of all traces of $\Sys$ under the labeling map $\Ll$. 
Then the system $\Sys$ satisfies an LTL property $\phi$ under labeling map $\Ll$ if for all $w \in TR(\Sys, \Ll)$, we have $w \models \phi$.
We denote this as $\Sys \models_{\Ll} \phi$ and infer the labeling map from context.
As LTL subsume safety and persistence, one can formulate these as LTL formulae.

\vspace{0.5em}\noindent \textbf{Nondeterminstic B\"uchi Automata.}
A nondeterminstic B\"uchi automaton (NBA) $\Aa$ is a tuple $(\Sigma,Q, Q_0, \delta, Acc)$, where:
\begin{itemize}
    \item $\Sigma$ is the alphabet, 
    \item $Q$ a finite set of states,
    \item $Q_0 \subseteq Q$ an initial set of states, 
    \item $\delta \subseteq Q \times \Sigma \times Q$ the transition relation, and 
    \item ${Acc} \subseteq Q$ denotes a set of accepting states.
\end{itemize}
A run of the automaton $\Aa = (\Sigma,Q, q_0, \delta, Acc)$ over a trace $w = \langle \sigma_0, \sigma_1, \sigma_2 \ldots, \rangle \in \Sigma^{\omega}$, is an infinite sequence of states  characterized as $\rho = \langle q_0,q_1, q_2, \ldots, \rangle \in Q^{\omega}$ with $q_0 \in Q_0$ and $q_{i+1} \in \delta(q_i, \sigma_i)$.
An NBA $\Aa$ is said to accept a trace $w$, if there exists a run $\rho$ over $w$ where $\Inf(\rho) \cap {Acc} \neq \emptyset$.

It is well known~\cite{vardi_2005_automata} that given an LTL formula $\phi$ over a set of atomic propositions $AP$, one can construct an NBA $\Aa$ such that $w \in L(\Aa)$ iff $w \models \phi$.
An automata-theoretic technique to determine whether $\Sys \models_{\Ll} \phi$ is to first find the NBA $\Aa$ that represents $\neg \phi$, and then ensure that $\Sys \not\models_{\Ll} \neg \phi$ by showing that no trace of the system is accepted by the NBA $\Aa$.
While converting an LTL formula to an NBA is exponential in the size of the formula, negating an LTL formula has a complexity that is linear in its size.

\section{Closure Certificates}

Podelski and Rybalchenko~\cite{podelski_2004_transition} introduced the notion of transition invariants as an over-approximation of the transitive closure of the transition relation, restricted to the set of reachable states. 
If the transition invariant restricted to pairs of accepting states is disjunctively well-founded, then they showed that no execution can visit these accepting states infinitely often, refuting the $\omega$-regular specification. 
In this section, we introduce closure certificates (CC) as a functional analog of \emph{transition invariants}.

Recall that barrier certificates are functional analogs to inductive state invariants in the following way:
all the initial states are within the zero level set of the barrier certificate, and, given any state that is within the zero level set, its successor according to the transition relation is also in the zero level set.
Our definition of closure certificates are a functional analog to inductive \emph{transition invaraints}. 
We study their use in the verification of safety, persistence (refuting recurrence), and LTL specifications.

\subsection{Closure Certificates for Safety}
\label{subsec:verif_safety}
We first define closure certificates for safety.
\begin{definition}[Closure Certificate for Safety]
\label{def:tbar_safe}
Consider a system  $\Sys = (\Xx, \Xx_0, f)$. A  function
    $\Tt: \Xx \times \Xx \to \R$ is a Closure Certificate (CC) for a set of unsafe states $\Xx_{u}$ if there exists a value $\xi \in \R_{ > 0}$ such that for all states $x, y \in \Xx$, $x' \in f(x)$, and states $x_0 \in \Xx_0$ and $x_u \in \Xx_u$, we have:
    \begin{align}
        & \big( \Tt(x, x') \geq 0 \big) \label{eq:tbar_cond_1_safe} \\
        & \big( \Tt(x', y) \geq 0 \big) \implies \big( \Tt(x, y) \geq 0 \big), \text{ and } \label{eq:tbar_cond_2_safe}\\
        & \big( \Tt(x_0, x_u) \leq - \xi \big). \label{eq:tbar_cond_3_safe} 
    \end{align}
\end{definition}

The existence of a closure certificate implies the safety of the system $\Sys = (\Xx, \Xx_0, f)$ as shown next.

\begin{theorem}[Closure Certificate imply Safety]
    \label{thm:closure_safe}
   Consider a system $\Sys$. The existence of a function $\Tt: \Xx \times \Xx \to \R$ that satisfies conditions~\eqref{eq:tbar_cond_1_safe}-\eqref{eq:tbar_cond_3_safe} implies its safety.
\end{theorem}
\begin{proof}
Let us assume that there exists a trace of the system $\langle x_0, \ldots, x_u , \ldots \rangle$ that reaches an unsafe state $x_u \in \Xx_u$ from some initial state $x_0$.
From condition~\eqref{eq:tbar_cond_1_safe}, we have 
$\Tt(x_i, x_{i+1}) \geq 0$ for all $i \in \N$, and from condition~\eqref{eq:tbar_cond_2_safe} and induction, we have $\Tt(x_0, x_i) \geq 0$ for all $i \in \N$.
Thus we must have $\Tt(x_0, x_u) \geq 0$ as $x_j = x_u$ for some $j \in \N$.
According to condition~\eqref{eq:tbar_cond_3_safe}, 
$\Tt(x_0, x_u) \leq - \xi$, where $\xi \in \R_{>0}$, which is in  contradiction to the previous inequality. 
\end{proof}

Observe that closure certificates are defined over pairs of states of the system rather than just over the states of the system. Hence, searching for a closure certificate suffers computationally more than a search for a barrier certificate.
On  the other hand, for a certain template of functions (e.g., linear or quadratic), one might be able to find closure certificates, even when barrier certificates of the same template do not exist.
In particular, we have the following result:
\begin{theorem}[Simplicity of Closure Certificates]
\label{thm:Bar_Tinv}
For any natural number $d \in \N$, there exists a system $\Sys $ with unsafe set of states $\Xx_u$ that cannot be shown to be safe by a polynomial barrier certificate of degree $d$ but can be shown to be safe by a linear closure certificate.
\end{theorem}

\begin{proof}
Consider a system $\Sys = (\Xx, \Xx_0, f)$, with $\Xx = [0, (2d+2)]$ as the state set,  $\Xx_0 = \set{1,3, \ldots, (2d + 1) } $ as the initial set of states, and a constant transition relation  $f(x) = \set{ 0 }$  for every state $x \in \Xx$.
Let the set of unsafe states be $\Xx_u = \set{2, 4, \ldots, (2d + 2)}$.
We observe that the system is trivially safe.

Let us suppose there exists a polynomial barrier certificate $\Bb: \Xx \to \R$ of degree $d$ that acts as a proof of safety.
From conditions~\eqref{eq:bar_cond_1} and~\eqref{eq:bar_cond_2}, we have $\Bb(x) \leq 0$ for every state $x \in \Xx_0$, and $\Bb(x) > 0$ for every state $x \in \Xx_u$.
Applying intermediate value theorem, the function $\Bb$ needs to change signs in at least $(d+1)$ points and must therefore have at least $(d+1)$ roots. 
This contradicts our assumption that $\Bb$ is a polynomial of degree $d$.

Consider the function $\Tt: \Xx \times \Xx \to \R$ defined as $\Tt(x,y) = -y$.
Observe that $0 \in f(x) $ for all $x \in \Xx$, and that $\Tt(x,y) \geq 0$ only when $y \leq 0$. 
This implies conditions~\eqref{eq:tbar_cond_1_safe} and~\eqref{eq:tbar_cond_2_safe} are satisfied.
Further for every state $x_u \in X_u$ and every state $x_0 \in X_0$, we have $\Tt(x_0, x_u) \leq -2 $.
Thus condition~\eqref{eq:tbar_cond_3_safe} also holds.
We  conclude that the function $\Tt$ is a closure certificate and acts as a proof that the system is safe.
\end{proof}
We should note that while the proof of the above Theorem relied on showing that no barrier certificate exists for a finite state system, one can employ similar techniques for a continuous space example. Consider the system $\Sys = (\Xx, \Xx_0, f)$, where $\Xx = \R$ denotes the state set, $\Xx_0 = \set{0,\frac{1}{4}, \ldots \frac{1}{2^{d+2}}}$ indicates the initial set of states, and $f(x) = \{ x+1 \}$ denotes the transition relation for every state $x \in \Xx$. Let the set of unsafe states be $\Xx_u = \set{\frac{1}{2}, \ldots, \frac{1}{2^{d+3} } }$, then there exists no polynomial function of degree $d$ that acts as a barrier certificate for the above function.
However the function $\Tt(x,y) = y - x - 1$ acts as a closure certificate that ensures the system starting from the initial state does not reach the unsafe set of states. An illustration of this example for degree $2$ can be found in Appendix~\ref{ap:thm_cont_BarTinv}.

Previously, we discussed how one can use closure certificates even when barrier certificates fail.
We now show that if a system can be guaranteed to be safe via barrier certificates, then it can be guaranteed via closure certificates as well.
\begin{theorem}[Expressiveness]
    \label{thm:bar_to_close}
    Consider a system $\Sys = (\Xx, \Xx_0, f)$, with unsafe set of states $\Xx_u$.
    Given a barrier certificate $\Bb: \Xx \to \R$ (Definition~\ref{def:bar}), one can compute a closure certificate $\Tt: \Xx \times \Xx \to \R$.
\end{theorem}
\begin{proof}
Let $\gamma \in \R_{ > 0}$.
     We define the function $\Tt: \Xx \times \Xx \to \R$  as:
     \[
     \Tt(x,y) = 
     \begin{cases}
     0 & \text{ if } \Bb(x) > 0  \text{ or }  \Bb(y) \leq 0, \\
     -\gamma & \text{otherwise}.
     \end{cases}
     \]
     We now show that $\Tt$ is a CC with $\xi = \gamma$.
     Let us suppose that $\Tt(x, x') < 0$ for some $x \in \Xx$.
     For this to be true, we must have $\Bb(x) \leq 0$, and $\Bb(x') > 0$, however, this contradicts condition~\eqref{eq:bar_cond_3} and so $\Tt$ must satisfy condition~\eqref{eq:tbar_cond_1_safe}.
     Second, suppose $\Tt(x', y) \geq 0$, and $\Tt(x,y) < 0$.
     Then $\Bb(x) \leq 0$, $\Bb(y) > 0$, and one of $\Bb(x') > 0$ or $\Bb(y) \leq 0$.
     Since both $\Bb(y) \leq 0$ and $\Bb(y) > 0$ cannot be true, we must have $\Bb(x) \leq 0$, and $\Bb(x') > 0$, which again contradicts condition~\eqref{eq:bar_cond_3}, and so condition~\eqref{eq:tbar_cond_2_safe} must hold.
     Finally, consider $\Tt(x_0, x_u)$. From conditions~\eqref{eq:bar_cond_1} and~\eqref{eq:bar_cond_2}, we have $\Bb(x_0) \leq 0$, and $\Bb(x_u) > 0$, and, hence, by definition $\Tt(x_0, x_u) = - \xi$ satisfies condition~\eqref{eq:tbar_cond_3_safe}. 
\end{proof}

\subsection{Closure Certificates for Persistence}
\label{subsec:finite_vis}
Similar to how closure certificates are used to guarantee safety, one may use  closure certificates to show a region is visited finitely often. 
This relies on showing that the closure certificate is well-founded, similar to the condition used in~\cite{podelski_2006_model}.

\begin{definition}[Closure Certificates for Persistence]
\label{def:tbar}
   Consider a system $\Sys = (\Xx, \Xx_0, f)$. A bounded    function $\Tt: \Xx \times \Xx \to \R$ is a Closure Certificate (CC) for $\Sys$ with set of states $\Xx_{VF} \subseteq \Xx$, that must be visited finitely often if there exists a value $\xi \in \R_{ > 0}$ such that for all states $x,y \in \Xx$,  $x' \in f(x)$,  $x_0 \in \Xx_0$, and all states $y', y'' \in \Xx_{VF}$ we have:
    \begin{align}
        &\big( \Tt(x, x') \geq 0 \big), \label{eq:tbar_cond_1} \\
        & \big( \Tt(x', y) \geq 0 \big) \implies \big( \Tt(x, y) \geq 0 \big) \text{ and } \label{eq:tbar_cond_2}\\
        &\big( \Tt(x_0, y') \geq 0 \big) \wedge \big( \Tt(y',y'') \geq 0 \big) \implies \nonumber \\
         &\big( \Tt(x_0,y'') \leq \Tt(x_0,y') - \xi \big).   \label{eq:tbar_cond_3} 
    \end{align}
\end{definition}
\begin{theorem}[Closure Certificates imply Persistence]
    \label{thm:tbar_visit_finite}
    Consider a system $\Sys$. The existence of a function $\Tt: \Xx \times \Xx \to \R$ that satisfies conditions~\eqref{eq:tbar_cond_1}-\eqref{eq:tbar_cond_3} implies that the traces of the system visit the set $\Xx_{VF}$ finitely often.
\end{theorem}
\begin{proof}
Let us suppose that there is some trajectory $ \langle x_0, x_1, \ldots, \rangle$ of the system that starts from state $x_0 \in \Xx_0$ and visits $\Xx_{VF}$ infinitely often.
Let the infinite sequence $\langle y_0, y_1, \ldots,  \rangle $ denote the states that are visited in $\Xx_{VF}$ in that order, \textit{i.e.}, the trajectory is $\langle x_0, \ldots, y_0, \ldots ,y_1, \ldots \rangle$.
From conditions~\eqref{eq:tbar_cond_1} and~\eqref{eq:tbar_cond_2}, we have $\Tt(x_0, y_i) \geq 0$  and  $\Tt(y_i, y_j) \geq 0$ for all indices $j > i$, $i,j \in \N$.
As we assume the function $\Tt$ to be bounded, there exists some $\Tt^* \in \R$, such that $\Tt(x,y) \leq \Tt^*$ for every pair of states $x,y \in \Xx$.
Note that $\Tt(x_0, y_0) \leq T^{*}$.
From condition~\eqref{eq:tbar_cond_3}, and induction, we have \[
\Tt(x_0, y_i) \leq \Tt(x_0, y_0) - i \xi \leq T^{*} - i \xi.
\]
As this is true for all $i \in \N$, and we have $\xi \in \R_{ > 0}$, there must exist some $j \in \N$ such that $\Tt(x_0, y_j) < 0$.
This is a contradiction. 
\end{proof}

\subsection{Closure Certificates for LTL Specifications}
\label{subsec:ltl_verif}
To verify whether the system satisfies a desired LTL formula $\phi$, we  first construct the NBA $\Aa = (Q, Q_0, \delta, Q_{Acc})$ that represents the complement of the specification $\neg \phi$.
Observe that the state set of the NBA is finite, and therefore we can denote the set $Q$ as the set $\{0, 1, \ldots, |Q| - 1 \}$.
We then construct the product $\Sys \otimes \Aa = (\Xx', \Xx_0', f')$ of the system $\Sys = (\Xx, \Xx_0, f)$ with the NBA $\Aa$, where: 
\begin{itemize}
    \item $\Xx' = \Xx \times \set{0, \ldots, |Q| - 1 }$ indicates the state set
    \item $\Xx'_0 = \Xx_0' \times \set{ q_0 \mid q_0 \in Q_{0}}$ indicate the initial set of states.
    \item the state transition relation $f'$ is defined as :
 \[ f'((x,q_i)) = \big\{ (x', q_j) \mid q_j \in \delta(q_i, \Ll(x)), \text{ and } x' \in f(x) \big\}. \]

\end{itemize}

To verify whether a given system satisfies a desired LTL property, we make use of a closure certificate on the product $\Sys \otimes \Aa$.

\begin{definition}[Closure Certificate for LTL]
\label{def:tbar_prod}
   Consider a system $\Sys = (\Xx, \Xx_0, f)$ and NBA $\Aa= (Q, Q_0, \delta, Acc)$ representing the complement of an LTL formula $\phi$. A bounded function $\Tt: \Xx \times Q \times  \Xx \times Q \to \R$ is a closure certificate for $\Sys$ and NBA $\Aa$  if there exists a value $\xi \in \R_{ > 0}$ such that
    for all states $x, y \in \Xx$, $x' \in f(x)$ and states $i, j \in Q$, and $i' \in \delta(i, \Ll(x))$, we have:
    \begin{align}
        & \Big( \Tt\big( (x, i), (x', i') \big) \geq 0 \Big) \label{eq:tbar_prod_cond_1} \\
        & \Big( \Tt \big( (x', i'), (y, j) \big) \geq 0 \Big) {\implies} \Big( \Tt \big((x, i), (y, j) \big) \geq 0 \Big) \label{eq:tbar_prod_cond_2} 
    \end{align}
    and for all states $x_0 \in \Xx_0$, $s \in Q_0$, and $\ell, \ell' \in Acc$, we have:
    \begin{align}
        &\Big( \Tt\big((x_0,s),(y,\ell) \big) \geq 0 \Big) \wedge \Big( \Tt( (y,\ell), (y',\ell') ) \geq 0 \Big) \implies \nonumber \\
        & \qquad \Big(\Tt \big( (x_0,s),(y',\ell') \big) \leq \Tt \big( (x_0,s),(y,\ell) \big) - \xi \Big). 
        \label{eq:tbar_prod_cond_3} 
    \end{align}
\end{definition}
Now, we provide the next result of the paper on the verification of LTL specifications using closure certificates on $\Sys \otimes \Aa$.
\begin{theorem}[Closure Certificates verify LTL]
\label{thm:tbar_omega}
   Consider a system $\Sys$ and an LTL formula $\phi$. Let NBA $\Aa$ represent the complement of the specification, \textit{i.e}, $\neg \phi$. The existence of a closure certificate satisfying conditions~\eqref{eq:tbar_prod_cond_1}-\eqref{eq:tbar_prod_cond_3} implies that $\Sys \models_{\Ll} \phi$. \end{theorem}
\begin{proof}
    Observe that a CC $\Tt$ satisfying conditions~\eqref{eq:tbar_prod_cond_1} to~\eqref{eq:tbar_prod_cond_3} is a CC for the product of $\Sys$ and $\Aa$.
    From Theorem~\ref{thm:tbar_visit_finite}, we observe that the product system  visits accepting states finitely often and so we infer that no trace of the system is in the language of the NBA $\Aa$.
    The proof is now complete.
\end{proof}
\section{Synthesizing Closure Certificates}
\label{sec:synth}
This section presents two approaches to synthesize closure certificates when the dynamical systems under study have state sets which are subsets of $\R^n$, \textit{i.e.}, $\Xx \subseteq \R^n$, and the transition function $f$ is a polynomial.
The first approach we consider is using  a counterexample guided approach via Satisfiability Modulo Theory (SMT) solvers~\cite{barrett_2018_satisfiability}, while the second makes use of standard sum-of-squares (SOS)~\cite{Parrilo_2003} approaches to find closure certificates similar to barrier certificates.
In the following sections we describe the relevant conditions for persistence and verifying $\omega$-regular objectives.
The conditions for safety can be recovered in a straightforward manner and are thus ommitted from the following discussion.

\subsection{SMT-based Approach}
\label{sec:CEGIS}
Counterexample-guided Inductive Synthesis (CEGIS)~\cite{Lezama_2008_thesis} has seen significant use in the synthesis of barrier certificates.
We thus consider conditions to provide a CEGIS approach to find closure certificates.
To find a CC as in Definition~\ref{def:tbar}, we first fix the template of the CC  to be a linear combination of user-defined basis functions:
\[
\Tt(x,y) = \sum_{m = 1}^{z} c_m p_m(x,y),
\]
where functions $p_m$ are user-defined  analytical basis functions over the state variables $x$ and $y$ and $c_1, \ldots, c_z$ are the coefficients.
As an example, we can consider $c_1, \ldots,  c_z$ to be real values, and $\Tt(x,y)$ to be a polynomial.
In such a case, the functions $p_1, \ldots, p_m$ are monomials over $x$ and $y$.
Note that if the values of $x, y \in \Xx$  are fixed, then the only decision variables in $\Tt(x,y)$ are the coefficients $c_m$, $m\in\{1,\ldots,z\}$.

We  sample $2N$ points from the state set $\Xx$ of the system to create the sets $S_1 = \set{ x_1, \ldots, x_N }$, and $S_2 = \set{y_{1}, \ldots, y_{N}}$, and sample $3N$ points from $\Xx_0$, $\Xx_{VF}$, and $\Xx_{VF}$, respectively, to create  sets $S_3 = \set{x_{0,1}, \ldots, x_{0,N}}$,  and $S_4 = \set{z_1, \ldots, z_{2N}}$, respectively.
We then encode the constraints of the closure certificate for every pair of points  as an SMT-query over the theory of linear real arithmetic (LRA)~\cite{dutertre_2006_fast} using z3~\cite{moura_2008_z3} as follows:
\begin{align}
& \underset{x \in S_1}{\bigwedge} \Big(  \Tt(x, x') \geq 0 \Big), \label{eq:smt_tbar_1} \\
& \underset{x \in S_1, y \in S_2}{\bigwedge} \Big( \big(\Tt(x', y) \geq 0 \big) \implies \big(\Tt(x,y) \geq 0 \big)  \Big),  \text{ and } \label{eq:smt_cbar_2} \\
& \underset{x_0 \in S_3, z, z' \in S_{4}}{\bigwedge} \Big( \big(\Tt(x_{0}, z) \geq 0 \big) \wedge \big(\Tt(z, z')\geq 0 \big) \\
& \qquad \implies \big(\Tt(x_{0}, z') \leq \Tt(x_{0}, z) - \xi \big) \Big),  \label{eq:smt_cbar_3} 
\end{align}
where $x' = f(x_{k_1})$ indicates the next state from $x_{k_1}$ following the transition function.
We lastly add a constraint of $\xi$ being larger than some small positive value and then find values $c_1, \ldots, c_z $  for the coefficients and substitute them as a candidate CC $\Tt(x,y)$.

To determine if this candidate is in fact a CC, we now try to find elements $x, y, x_0,z, z' \in \Xx$ such that one of the conditions~\eqref{eq:tbar_cond_1}-\eqref{eq:tbar_cond_3} does not hold.
We do this by encoding the negation of these conditions as an SMT query.
If such a counterexample is found, we add them to the respective set and repeat the process.
If no counterexample is found, then we conclude that this is a CC.

Instead of using an SMT solver to find a  candidate CC, we can instead run our CEGIS loop quicker by strengthening conditions~\eqref{eq:tbar_cond_2}-\eqref{eq:tbar_cond_3} as inequalities of the form: 
    \begin{align}
        & \tau_1 \Tt(x', y) \leq  \Tt(x, y), \label{eq:tbar_cond_2_ineq}\\
        &\Tt(x_0,y) \!-\! \xi \!-\!\Tt(x_0,y') \!\geq\!  \tau_2 \Tt(x_0, y) \!+\! \tau_3\Tt(y, y'), \label{eq:tbar_cond_3_ineq}
    \end{align} 
    for all states $x_0 \in \Xx_0$, and $y, y' \in \Xx_{VF}$.  
where $\tau_1, \tau_2, \tau_3 \in \R_{\geq 0}$ are fixed nonnegative values.
The satisfaction of conditions~\eqref{eq:tbar_cond_2_ineq} and~\eqref{eq:tbar_cond_3_ineq} implies the satisfaction of conditions~\eqref{eq:tbar_cond_2} and~\eqref{eq:tbar_cond_3}, and the search for a candidate CC can be cast as a linear program instead. 
This allows one to use a linear programming solver (such as Gurobi~\cite{gurobi}) to find a candidate CC instead.
We then find a counterexample via SMT queries similar to the earlier approach, and then add the counterexample to our linear program, and search for a candidate CC again.
While conditions \eqref{eq:tbar_cond_2_ineq} and~\eqref{eq:tbar_cond_3_ineq} are more conservative, the search for a candidate is much quicker.

We adopt a similar approach to find a CC for the synchronized product as in Definition~\ref{def:tbar_prod}, that acts as a proof that the traces of the system satisfy an LTL property whose negation is specified by the language of an NBA $\Aa = (\Sigma, Q, Q_0, \delta, Acc)$.
{In this setting, we assume our closure certificates to be piecewise with respect to pairs of states of NBA $\Aa$.
Each piecewise component is then considered to be a linear-combination of some user-defined basis functions.
For every pair of states $i,j \in Q$, we denote the corresponding piecewise component as $\Tt_{i,j}$.
We define each piecewise component as:  \[\Tt_{i,j} (x, y) = \sum_{m = 1}^{z} c_{m.i.j} p_{m,i,j}(x,y), \] where the functions $p_{m,i,j}$ are user-defined basis functions over the states $x, y \in \Xx$, and $c_{m,i,j}$ are the coefficients. 
}
We then encode the constraints as the following conjunctions for every state $x \in S_1$, $y \in S_2$, $x_{0} \in S_3$ and $z,z' \in S_2$ as well as every state $i, {j} \in Q$ such that ${i'} \in \delta({i}, \Ll(x))$, and states $s \in Q_0$ and $\ell, \ell' \in {Acc}$: 
\begin{align}
& \underset{x \in S_1}{\bigwedge} \Big(  \Tt_{i, i'}(x, x' ) \geq 0 \Big), \label{eq:smt_tprod_1} \\
& \underset{x \in S_1, y\in S_2}{\bigwedge} \Big(  \big(\Tt_{i',j}(x,y) \geq 0\big) {\implies} \big( \Tt_{i, j}(x, y) \geq 0  \big)  \Big) \label{eq:smt_tprod_2}, \text{ and }\\
& \underset{x_0 \in S_3, z,z' \in S_{2}}{\bigwedge} \Big(  \big(\Tt_{s,\ell}(x_0, z) \geq 0 \big) \wedge \big( \Tt_{\ell, \ell'}(y, z')\geq 0 \big)  && \nonumber \\
& \qquad \implies \big( \Tt_{s, \ell'}(x_0, z') \leq \Tt_{s,\ell}(x_0, z ) - \xi  \big) \Big).\label{eq:smt_cprod_3} 
\end{align}
In general there is no guarantee of termination when using a CEGIS approach for uncountable state sets.
However, one may strengthen the conditions as specified in~\cite{kong_2018_delta} to guarantee termination of the CEGIS loop.
Instead of using a CEGIS approach, one may also encode the conditions in an SMT solver over the nonlinear theory of reals~\cite{gao_complete_2012} such as dReal~\cite{gao_dreal_2013} or z3~\cite{jovanovic_2013_solving} to search for CCs.
While all the above approaches are NP-hard~\cite{gao_dreal_2013,collins_1975_quantifier,jovanovic_2013_solving}, we find the CEGIS approach to work better in practice compared to searching for a solution in the nonlinear theory of reals.

Note that barrier certificates face many of the same challenges when using a CEGIS approach.
Computationally, however, closure certificates take more time in practice as these are defined over pairs of states rather than over a single state, and so suffer more when the dimension of the state set increases.
 We should add that we have not considered the complexity for finding the NBA representing the complement of the specification, but rather assume this NBA to be given.
While the complexity of NBA complementation is EXPTIME~\cite{safra_1988_complexity}, it takes linear time to complement an LTL formula.
However converting an LTL formula to an NBA has exponential complexiy in the size of the formula~\cite{vardi_2005_automata}.

\subsection{Sum-of-Squares based Approach}
\label{subsec:SOS}
The technique of using semidefinite programming ~\cite{Parrilo_2003} and casting the search for a standard barrier certificates~\cite{prajna_2004_safety} as  SOS  polynomials is particularly important due to the simpler complexity of computation when compared to CEGIS approaches.
We show how one may adopt a SOS approach to find closure certificates.
To do so, we first note that a set $A \subseteq \R^n$ is semi-algebraic if it  can be defined with the help of a vector of polynomial inequalities $h(x)$ as $A = \{ x \mid h(x) \geq 0 \}$, where the inequalities is interpreted component-wise.

To adopt a SOS approach to find CCs as in Definition~\ref{def:tbar}, we consider the sets $\Xx$, $\Xx_0$, and $\Xx_{VF}$ to be semi-algebraic sets defined with the help of vectors of polynomial inequalities $g_{A}$, $g_0$, and $g_{VF}$, respectively.
As these sets are semi-algebraic, the sets $\Xx \times \Xx$ and $\Xx_0 \times \Xx_{VF} \times \Xx_{VF}$ are semi-algebraic as well.
Let their corresponding vectors be $g_{B}$ and $g_{C}$, respectively.
Furthermore, we assume that the user-defined basis functions $p_m$ are monomials and again strengthen the implications in conditions~\eqref{eq:tbar_cond_2}-\eqref{eq:tbar_cond_3} to conditions~\eqref{eq:tbar_cond_2_ineq}-\eqref{eq:tbar_cond_3_ineq}.
Then the search for a CC as in Definition~\ref{def:tbar} reduces to showing that the following polynomials are sum-of-squares:
\begin{align}
  & \Tt(x,x')-\lambda_{A}^T(x)g_{A}(x),  \label{eq:sos_tbar_1} \\
  &\Tt(x,y) - \tau_1\Tt(x',y)-\lambda_{B}^T(x,y)g_{B}(x,y), \text{ and }\label{eq:sos_tbar_2} \\
   & \Tt(x, y') - \xi - \tau_2\Tt(x,y)  \nonumber \\&-\tau_3\Tt(y,y') - \lambda_{C}^T(x,y,y')g_{C}(x,y,y'),  \label{eq:sos_tbar_3} 
\end{align}
where $x' = f(x)$, the multipliers $\lambda_{A}$, $\lambda_{B}$, $\lambda_{C}$, are sum-of-squares over the state variable $x$, the state variables $x,y$, and the state variables $x,y,y'$ over the sets $\Xx$, $\Xx \times \Xx$, and $\Xx_0 \times \Xx_{VF} \times \Xx_{VF}$ respectively, and  $\xi$, $\tau_1$, $\tau_2$, and $\tau_3 \in\R_{ > 0}$ are positive values.

\begin{lemma}
Assume the sets $\Xx$, $\Xx_{0}$, and $\Xx_{VF}$ are semi-algebraic, and there exists a sum-of-squares polynomial $\Tt(x,y)$ satisfying conditions~\eqref{eq:sos_tbar_1}-\eqref{eq:sos_tbar_3}.
Then the function $\Tt(x,y)$ is a CC  satisfying conditions~\eqref{eq:tbar_cond_1}-\eqref{eq:tbar_cond_3}.
\end{lemma}

Since there are finitely many letters $\sigma \in \Sigma$, without loss of generality, one can partition the set $\Xx$ into finitely many partitions $\Xx_{\sigma_1}, \ldots, \Xx_{\sigma_p}$, where for all $x \in \Xx_{\sigma_m}$, we have $\Ll(x) = \sigma_m$.
Given an element $\sigma_m \in \Sigma$, we can uniquely characterize the relation $\delta_{\sigma_i}$ as $(q'_i, q_i) \in \delta_{\sigma_i}$ if and only if  $q'_i \in \delta(q_i, \sigma_i)$.
Assume that the sets $\Xx$, $\Xx_{0}$, and $\Xx_{\sigma_m}$, for all $\sigma_m$, are semi-algebraic and characterized by polynomial vectors of inequalities $g(x)$, $ g_0(x)$, and $g_{\sigma_m, A}(x)$, respectively.
Furthermore, we consider polynomial vectors of inequalities $g_{(\sigma_m), B}(x, y)$ over the product space $\Xx \times \Xx$, and $g_{(\sigma_m, C)}(x, y, y')$ over $\Xx_0 \times \Xx \times \Xx$, respectively.
Let  the state transition function $f:\Xx \to \Xx$ be a polynomial function. Now, one can reduce the search for a CC to showing that the following polynomials are SOS for all states $x, y, y' \in \Xx$, $x' = f(x)$, and $x_0 \in \Xx_0$, and states $i, j \in Q$, $s \in Q_0$, and $\ell, \ell' \in {Acc}$, and letters $\sigma_m \in \Sigma$, such that $i' \in \delta_{\sigma_m} (i)$:
\begin{align}
&\Tt_{i',i}(x, x')-\lambda^T_{\sigma_m, A}(x)g_{\sigma_m, A}(x),   \label{eq:sos_1} \\
&\Tt_{i', j}(x, y) - \tau_1\Tt_{i, j}(x', y) \nonumber \\
&-\lambda^T_{\sigma_m, B}(x, y) g_{\sigma_m, B}(x, y), \text{ and }  \label{eq:sos_2}  \\
 & \Tt_{s, \ell'}(x_0, y') - \xi - \tau_2 \Tt_{s, \ell}(x_0, y) - \tau_3 \Tt_{\ell,\ell'}(y, y')\nonumber\\
  &- \lambda^T_{\sigma_m, C}(x_0, y, y') g_{\sigma_m,C}(x_0, y, y'), \label{eq:sos_3}     
 \end{align}
 are sum-of-squares, where $\lambda^T_{\sigma_i, A}$, $\lambda^T_{\sigma_i, B}$, and $\lambda^T_{\sigma_i, C}$ are sum-of-squares polynomials over their respective regions and $\tau_1, \tau_2, \tau_3, \xi \in \R_{> 0}$ are positive values.

\begin{lemma}
Assume the sets $\Xx$, $\Xx_{0}$, $\Xx_{VF}$, and $\Xx_{\sigma_i}$ for all $\sigma_i$ are semi-algebraic, and there exists sum-of-squares polynomials $\Tt_{i,j}(x,y)$ satisfying conditions~\eqref{eq:sos_1}-\eqref{eq:sos_3} for every $i, j \in Q$.
Then the function $\Tt \big( (x, i), (y, j) \big)$ defined piecewise as $\Tt_{i, j}(x,y)$ for all $i, j \in Q$ is a CC for the product satisfying conditions~\eqref{eq:tbar_prod_cond_1}-\eqref{eq:tbar_prod_cond_3}.
\end{lemma}

To determine whether the above equations are SOS, one can make use of solvers such as~\cite{prajna_2002_introducing}.
The complexity of determining whether the above equations are SOS is $O\big(\binom{2n+d}{d} \times \binom{2n+d}{d}\big)$, when searching for CCs for safety or ensuring finite visits, where $n$ is the dimension of the state set, and $2d$ is the degree of the polynomial.
The complexity of verifying LTL specifications is\footnote{Determining whether a polynomial in $n$ variables and degree $d$ are SOS can be reduced to a semidefinite program in $O\big(\binom{n+d}{d} \times \binom{n+d}{d}\big)$ variables \cite{Parrilo_2003}.} polynomial in $O\big( {2^{|\phi|}}^{2} \times 
\binom{2n+2d}{d} \times \binom{2n+2d}{d}\big)$, where $|\phi|$ indicates the size of the LTL formula.
This is because the closure certificate is a function of pairs of the state set of the system and there are at most $|Q|^2$ many pairs of transitions in an automaton.
The number of states of the NBA is $O\big(2^{|\phi|} \big)$, where $|\phi|$ is the size of the formula.
On the other hand, the complexity of determining whether the equations for barrier certificates are SOS is polynomial in $O\big(\binom{n+d}{d} \times \binom{n+d}{d}\big)$~\cite{prajna_2004_safety}.
If the dimension of the system is fixed, then the complexity is polynomial in the degree $2d$ but exponential in the size of the formula $\phi$.
The key issue when using an SOS approach, however, is that there may be polynomials that satisfy the above constraints but are not SOS.
Furthermore, one cannot directly encode the implication in SOS, and, hence, suffers from the conservatism of having to satisfy a stronger condition.

\section{Subsuming Existing Approaches} 
\label{sec:subsumption}

We show that our approach generalizes the existing class of techniques using state triplet introduced in~\cite{wongpiromsarn_2015_automata} for the verification of continuous-space systems against linear temporal logic properties.
The state triplet technique has been used for the verification and synthesis for stochastic systems~\cite{jagtap_2018_temporal,jagtap_2020_formal}, for networks of systems~\cite{anand_2022_small,anand_2021_compositional}, and in motion-planning for nonlinear systems~\cite{he_2020}.
Here, the transition map is a function, and the state set is a subset of $\R^n$.
First, we present the details of the state triplet approach briefly, in Section~\ref{subsec:triplet}. Then in 
Section~\ref{subsec:subsumption}, we show how one can use closure certificates  to guarantee satisfaction of LTL properties when the state triplet approach provides a guarantee as well.

\subsection{The State Triplet Approach}
\label{subsec:triplet}

Consider a system $\Sys = ( \Xx, \Xx_0, f)$, where $\Xx \subseteq \R^n$, and $f$ is a state transition function.
Consider a NBA $\Aa = (\Sigma, Q, Q_0, \delta, Acc)$ that represents the complement of the desired LTL formula $\phi$, and a labeling function $\Ll: \Xx \to \Sigma$.
The key idea of the state triplet approach is to find barrier certificates between edges of the automaton to disallow the system from visiting an accepting state.
This ensures that $\Sys \not\models_{\Ll} \neg \phi$, and so we have $\Sys \models_{\Ll} \phi$.

The steps of the approach are as follows:
\begin{enumerate}
    \item Consider all the simple paths in the NBA that start from an initial state and reach an accepting state.
    \item Break these paths into a sequence of state triplets $(q_m, q'_m, q''_m)$ (or edge pairs $(e_m, e'_m)$).
    \item Search for a barrier certificate to ``cut'' at least one triplet from each path.
    \item If we can cut at least one triplet along each path, we can conclude that $\Sys \models_{\Ll} \phi$, and if not this approach is inconclusive.
\end{enumerate}

To help illustrate this approach consider a system $\Sys = (\Xx, \Xx_0, f)$, and a finite alphabet $\Sigma = \set{a_0, a_1}$; the labeling map $\Ll$ naturally partition the set $\Xx$ into two sets $\Xx_{a_1}$ and $\Xx_{a_2}$.
 Let the NBA $\Aa = ( \Sigma,Q, Q_0, \delta, {Acc})$ in Figure~\ref{fig:aut_eg_1} represent the complement of an LTL specification of interest, where  $\Sigma = \{a_0, a_1\}$,  $Q = \{q_0, q_1, q_2, q_3\}$,  $Q_0 = \{q_0 \}$,  ${Acc} = \{ q_3 \}$, and $\delta$ is specified by the edges in the graph.
This NBA accepts those words which start with a $a_1$ and have at least two $a_0$'s in them.
It represents the LTL formula $\phi = a_1 \wedge \mathsf{X} \mathsf{F} (a_0 \wedge \mathsf{F} a_0)$.
There is one simple path starting from the initial state $q_0$ that reaches the accepting state $q_3$.
This path corresponds to the sequence of states $(q_0, q_1, q_2, q_3)$ and can be broken into two triplets $(q_0, q_1, q_2)$ and $(q_1, q_2, q_3)$.
The first state triplet corresponds to the edge pair $((q_0, q_1) , (q_1, q_2) )$ which are labeled by the pair of letters $(b,a)$, and the second to the edge pair $((q_1, q_2), (q_2, q_3))$ which are labeled by the pair of letters $(a_0, a_0)$.

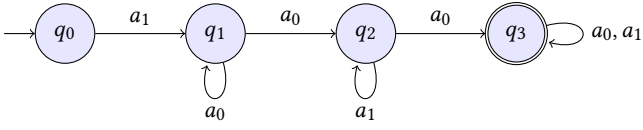
\begin{figure}[t]
    \centering
    \begin{tikzpicture}[node distance =2cm]
    \node[initial, state, draw, initial text =,fill=blue!10!white] (0) at (0,0) {$q_0$};
    \node[,state, fill=blue!10!white,] (1) at (2,0) {$q_1$};
    \node[,state, fill=blue!10!white,] (2) at (4,0) {$q_2$};
    \node[accepting,state, fill=blue!10!white,] (3) at (6,0) {$q_3$};    
    \path[->]
    (0) edge node[above]{$a_1$} (1)
    (1) edge[loop below] node{$a_0$} (1)
    (1) edge node[above]{$a_0$} (2)
    (2) edge node[above]{$a_0$} (3)
    (2) edge[loop below] node{$a_1$} (2)
    (3) edge[loop right] node{$a_0 , a_1$} (3);
    \end{tikzpicture}
    \caption{Example NBA $\Aa$ which represents the complement of a safety to illustrate the state triplet approach.}
    \label{fig:aut_eg_1}
\end{figure}

To cut the transitions along the first state triplet $(q_0, q_1, q_2)$, we try to find a barrier certificate, where the initial set of states are all the states with the label $a_1$ (corresponding to the edge $(q_0, q_1)$), \textit{i.e.} all the states of the system in $\Xx_{a_1}$. 
The unsafe states are all the states with a label of $a_0$ (corresponding to the edge $(q_1, q_2)$), \textit{i.e.} all the states in the set $\Xx_{a_0}$.
The existence of a barrier certificate, proves that no trace of the system can visit a state with label $a_0$, after visiting a state with label $a_1$, and so cannot correspond to the run $(q_0, q_1, q_2)$ in the automaton.
This cuts the path from the initial state $q_0$ to the accepting state $q_3$ of the automaton and shows that no trace of the system can take this corresponding path in the NBA.
As there are no other simple paths to the accepting state, we conclude that no trace of the system is in the language of the NBA.
If we fail to find a barrier certificate for the first triplet, we then search for a barrier certificate in the next triplet $(q_1, q_2, q_3)$.
As the edges of the states in the triplet have the same label $a_0$, we cannot find a barrier certificate where the initial set and unsafe set are both $\Xx_{a_0}$.
If we fail to find a barrier certificate for both triplets, then our approach is inconclusive.

As the state triplet approach proves that no trace of the system can reach the accepting state, one expects that it can be leveraged in a similar fashion to bounded model checking.
Ideally unrolling the automaton for $k$-steps would allow one to verify that no trace of the system visits the accepting state more than $k$ times.
Unfortunately, this is not true, and the state triplet approach does not benefit when one unrolls more than once.

\begin{lemma}
\label{lem:triplet_twice}
    Consider a NBA $\Aa = (\Sigma, Q, Q_0, \delta, Acc)$, whose simple paths from the initial states $Q_0$ to the accepting states ${Acc}$ have been divided into $k$ state triplets $(q_m, q'_m, q''_m)$ for all $1 \leq m \leq k$, such that $q'_m \in \delta(q, a_m)$ and $q''_m \in \delta(q,b_m)$, 
    for some $a_m, b_m \in \Sigma$.
     Unrolling the automaton more than once does not lead to finding a new triplet with labels that have not been considered before.
\end{lemma}

We present the proof of Lemma~\ref{lem:triplet_twice} in Appendix~\ref{ap:proof_lem_twice}.
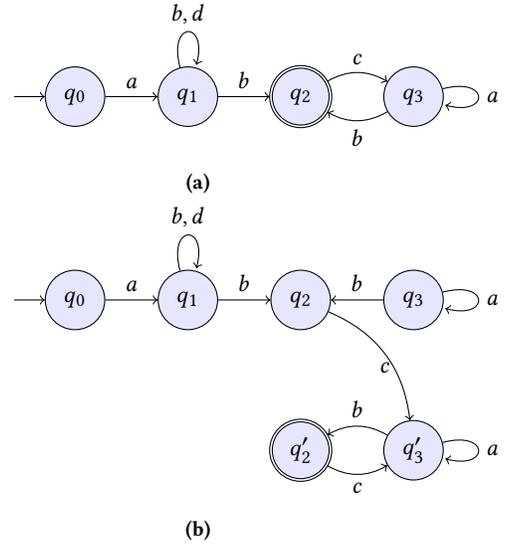
\begin{figure}[t]
    \begin{subfigure}{0.3\textwidth}
    \centering
    \begin{tikzpicture}[node distance =1cm]
    \node[initial, state, draw, initial text =,fill=blue!10!white] (0) at (0,0) {$q_0$};
    \node[,state, fill=blue!10!white,] (1) at (1.5,0) {$q_1$};
    \node[ accepting,state, fill=blue!10!white,] (2) at (3,0) {$q_2$};
    \node[state, fill=blue!10!white,] (3) at (4.5,0) {$q_3$};
    \path[->]
    (0) edge node[above]{$a$} (1)
    (1) edge[] node[above]{$b$} (2)
    (2) edge[bend left] node[above]{$c$} (3)
    (3) edge[bend left] node[below]{$b$} (2)
    (1) edge[loop above] node{$b,d$} (1)
    (3) edge[loop right] node{$a$} (3);
    \end{tikzpicture}
    \caption{}
    \label{subfig:aut_unroll}
    \end{subfigure}
      \begin{subfigure}{0.3\textwidth}
          \centering
   \begin{tikzpicture}[node distance =2cm]
    \node[initial, state, draw, initial text =,fill=blue!10!white] (0) at (0,0) {$q_0$};
    \node[,state, fill=blue!10!white,] (1) at (1.5,0) {$q_1$};
    \node[,state, fill=blue!10!white,] (2) at (3,0) {$q_2$};
    \node[state, fill=blue!10!white,] (3) at (4.5,0) {$q_3$};
    \node[ accepting,state, fill=blue!10!white,] (5) at (3,-2) {$q'_{2}$};
    \node[state, fill=blue!10!white,] (6) at (4.5,-2) {$q'_{3}$};
    \path[->]
    (0) edge node[above]{$a$} (1)
    (1) edge node[above]{$b$} (2)
    (2) edge[bend left] node[below]{$c$} (6)
    (3) edge node[above]{$b$} (2)
    (1) edge[loop above] node[]{$b,d$} (1)
    (3) edge[loop right] node[right]{$a$} (3)
    (6) edge[loop right] node{$a$} (6)
    (6) edge[bend right] node[above]{$b$} (5)
    (5) edge[bend right] node[below]{$c$} (6);
    \end{tikzpicture}
    \caption{}
    \label{subfig:aut_unrolled}
    \end{subfigure}
    \caption{Example NBA $\Aa$ (Figure~\ref{subfig:aut_unroll}) and its unrolling (Figure~\ref{subfig:aut_unrolled}) from Section~\ref{subsec:triplet}. }   \label{fig:aut_triplet}
\end{figure}

Now, with an example, we demonstrate why unrolling the automaton once might help in finding state triplets.
Consider the NBA $\Aa = (\Sigma, Q, Q_0, \delta, {Acc})$ in Figure~\ref{subfig:aut_unroll}, with $\Sigma =\{ a,b,c,d \}$, $Q = \{q_0, q_1, q_2, q_3 \}$, ${Acc} = \{ q_2 \}$, and the transition relation specified by the edges in the graph.
We unroll this automaton to get the automaton $\Aa'$ in Figure~\ref{subfig:aut_unrolled}.
Unrolling the automaton once allows us to consider the triplet $(q_2, q'_3,q'_2) $ whose edge labels correspond to the pair of letters $(c, b)$.
Observe that no state triplet in the original NBA corresponds to this pair of letters. 
Thus even if one was not able to find a barrier certificate for a state triplet in NBA $\Aa$ (for the state triplet $(q_0, q_1, q_2)$), one may still find a barrier certificate for a state triplet in NBA $\Aa'$ (the state triplet $(q_2, q'_3, q'_2)$).
Hence, no trace of the system can visit the accepting state more than once.

We observe that unrolling once does have an impact since we can now consider those state triplets along the simple cycles of the NBA. 
Unfortunately, unrolling twice does not help.
Thus, one is unable to verify those traces of a system which reach and cycle on accepting states more than twice, even if they visit accepting states finitely often.

The state triplet approach is conservative in the following direction: \emph{independently of the state runs in the automaton and of the initial states of the system $\Sys$,  one is required to break the edge pairs of every simple path regardless of what states of the automaton may be encountered before or after}.

\subsection{CC Subsumes State Triplet Approach}
\label{subsec:subsumption}
We now show that our approach generalizes the earlier state triplet one.
Consider a NBA $\Aa = (\Sigma, Q, Q_0, \delta, Acc)$, and let us assume that there exist $k$ barrier certificates $\Bb_1, \Bb_2 \ldots, \Bb_k$ associated with state triplets $(q_m, q_m', q_m'')$, (or edge pairs $(e_m, e'_m)$) for each $1 \leq m \leq k$, that act as a proof that every trace of the system is in $L(\Aa)$.
Furthermore, let $(a_m, b_m)$ be the pairs of letters associated with these triplets.
We divide the states of the NBA $\Aa$  into two sets $Q_l$, and $Q_r$.
A state $q \in Q$ is in the set $Q_l$ if It is not the middle element of a state triplet and there is a path from $q$ to the middle element of \emph{some} state triplet.
A state $q \in Q$ is in the set $Q_r$ if there is a path from the middle element of every state triplet to $q$.
The only states that are in neither of the sets are middle elements of the triplets.
As the state triplet approach ``cuts'' the transitions of the NBA, we observe that no trace of the system starting from a state in $Q_l$ can reach a state in $Q_r$.
Furthermore, we note that every state $s \in Q_0$ is in the set $Q_l$, and every state $\ell \in {Acc}$ is in the set $Q_r$.
We now show how one may use closure certificates to provide guarantees of satisfaction for LTL specifications when the state triplet approach can guarantee the satisfaction of the same specification.

\begin{theorem}
\label{thm:subsumption}
    Consider a System $\Sys = (\Xx, \Xx_0, f)$, labeling map $\Ll: \Xx \to \Sigma$, and NBA $\Aa = (\Sigma, Q, Q_0, \delta, Q_{Acc}) $ representing the complement of a desired LTL formula $\phi$. 
    Suppose that there exists barrier certificates $\Bb_1, \ldots, \Bb_k$ that show $ \Sys \models_{\Ll}  \phi$ via the state triplet approach. 
    Then there exists a closure certificate $\Tt$ that also acts as a proof that  $\Sys \models_{\Ll} \phi$.
\end{theorem}
\begin{proof}[Proof (sketch)]
We construct a closure certificate $\Tt: \Xx \times Q \times \Xx \times Q \to \R $ such that,  for $x,y \in \Xx$, and $i, j \in Q$, we have that $ \Tt(x,i, y,j) \geq 0$, if:
\begin{itemize} 
    \item $i \in Q_r$;
    \item $i \in Q_l$ and $j \in Q_l$;
    \item  $i$ is the middle element of triplet $m$, and $\Bb_m(x) \leq 0$, $j \in Q_l$ and $\Bb_m(y) \leq 0$;
    \item $i$ is the middle element of triplet $m$, and $\Bb_m(x) > 0$; or
    \item $i \in Q_l$, $j$ is the middle element of some triplet $m$, and $\Bb_m(y) \leq 0$.
\end{itemize}
Moreover, $ \Tt(x,i,y,j) < 0$, otherwise. 
This certificate clearly guarantees that no trace of the system can reach the accepting states.
A detailed proof is given in Appendix~\ref{ap:subsumption_proof}.
\end{proof}

\section{Case Studies}
We experimentally demonstrate the utility of closure certificates on Kuramoto oscillators and a two-room temperature model.
In the first example, we consider the problem of safety verification. Here, we show that we can verify the safety a $1$ dimensional Kuramoto oscillator via a linear closure certificate when a linear barrier certificate cannot do the same.
We then verify the safety of a $2$ dimensional Kuramoto oscillator by converting the safety objective to an LTL specification.
We then search for a closure certificate over the product of the system and the NBA representing the complement of the specification.
In the second example, we consider the problem of verifying the persistence of a two-room temperature model.
To do so, we convert the objective to an LTL specification, after which we search for a closure certificate over the product of the system and the NBA.

\label{sec:case_studies}
\subsection{Kuramoto Oscillator}
Kuramoto model~\cite{GUO2021106804} has been used widely to describe chemical and biological oscillators, with applications in neuroscience and modern power system analysis.
As a first case study, we consider a system $\Sys = (\Xx, \Xx_0, f)$ to model a Kuramoto oscillator whose dynamics are taken from~\cite{anand_2022_small}, where $\Xx = [0, 2\pi]$ indicates the state set, $\Xx_0 = [\frac{4\pi}{9}, \frac{5\pi}{9}]$ the initial set of states, and $\Xx_u = [\frac{7\pi}{9}, \frac{8\pi}{9}]$ denotes the unsafe set of states.
The transition function $f$ is defined as:
\[ f(x) = x + \tau \Omega + t_s K \sin(-x) -0.532x^2 + 1.69,  \]
where $x \in \Xx$ indicates the phase of the oscillator, $t_s = 0.1$ is the sampling time, $\Omega = 0.01$ is the natural frequency, and $K = 0.0006$ is the coupling strength.

We then search for a linear closure certificate as in Defintion~\ref{def:tbar_safe} to ensure the safety of the system.
To do so, we strengthen the implication in condition~\eqref{eq:tbar_cond_2_safe} to condition~\eqref{eq:tbar_cond_2_ineq}, with $\tau_1 = 1$, and sample $50$ points from the initial, unsafe, and entire state set.
We then solve a linear program to find a candidate closure certificate.
As z3~\cite{moura_2008_z3} cannot handle the function $\sin(-x)$, we instead use the solver dReal~\cite{gao_dreal_2013} to find counterexamples.
We add these counterexamples to the set of samples and repeat the procedure until we find no counterexamples.
We find the closure certificate $\Tt(x,y) = 10 - 4.094y$ that acts as a proof of safety.
The time taken for our CEGIS loop to terminate is around $10$ minutes on a machine running MacOS 11.2 (Intel i9-9980HK with $64$ GB of RAM).
We should note that the linear program encoding the barrier certificate conditions is infeasible when we consider a linear barrier certificate rather than a CC.

{
We now cast the problem of safety verification as a problem of verifying a system against the LTL formula $\mathsf{G} \neg a$, over the set of atomic propositions $AP = \{ a\}$, where a state is marked with label $\{ a\}$ if it is unsafe.
The complement of this specification is $\mathsf{F} a$, and the NBA $\Aa$ representing this is described in Figure~\ref{fig:aut_case_study_1_safe}.
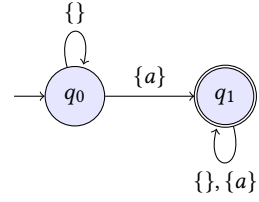
\begin{figure}[t]
    \centering
    \begin{tikzpicture}[node distance =2cm]
    \node[initial, state, draw, initial text =,fill=blue!10!white] (0) at (0,0) {$q_0$};
    \node[accepting,state, fill=blue!10!white,] (1) at (2,0) {$q_1$};
    \path[->]
    (0) edge[loop above] node[above]{$\{ \}$} (0)
    (0) edge[] node[above]{$\{ a\}$} (1)
    (1) edge[loop below] node[below]{$\{ \}, \{a \}$} (1);
    \end{tikzpicture}
    \caption{A (nondeterministic) B\"uchi automaton $\Aa$ representing the LTL formula $\mathsf{F} a$.}
    \label{fig:aut_case_study_1_safe}
\end{figure}
}
{
We consider the system $\Sys = (\Xx, \Xx_0, f)$ to be a two-dimensional Kuramoto oscillator, where $\Xx = [0, \frac{8\pi}{9}] \times [0, \frac{8\pi}{9}]$ denotes the state set.
$\Xx_0 = [0, \frac{\pi}{9}] \times [0, \frac{\pi}{9}]$ denotes the initial set of states and the transition function $f$ is defined as:
\[ f(x_1, x_2) =  \begin{bmatrix} x_1 \\ x_2  \end{bmatrix} +  \begin{bmatrix} \tau\Omega + 1.69 \\ \tau\Omega + 1.69  \end{bmatrix} + K t_s \begin{bmatrix} sin(x_2 - x_1) \\ sin(x_1 - x_2) \end{bmatrix} -0.532\tau\begin{bmatrix} x_1^2 \\ x_2^2  \end{bmatrix}, \]
where $(x_1, x_2) \in \Xx$ indicates the phase of the oscillators, and the remaining constants have the same values as the one dimensional case. 
We consider the alphabet $\Sigma = \{ \{ \}, \{ a \} \}$, and the labeling function $\Ll$ as $\Ll(x_1, x_2) = \{a \}$ if either $x_1 \in [\frac{15\pi}{18}, \frac{8\pi}{9}]$ or $x_2 \in [\frac{15\pi}{18}, \frac{8\pi}{9}]$.
All the other states are assigned a label of the empty set $\{ \}$.
We consider the template of the piecewise components of the closure certificate to be: 
\begin{multline} 
\Tt_{i,j} \big((x_1, x_2), (y_1, y_2) \big) = c_{0,i, j} + c_{1,i, j} y_1   \mathbb{I}_0(x_1, x_2)  + c_{2, i, j}   y_2   \mathbb{I}_0(x_1, x_2) \\ +  c_{3, i, j}   y_1   \mathbb{I}_a(x_1, x_2) + c_{4,i, j}   y_2   \mathbb{I}_a(x_1, x_2) + c_{5,i, j}   y_1 + c_{6, i, j}   y_2,
\end{multline}
for all states $(x_1, x_2), (y_1, y_2) \in \Xx$ and  NBA states  $i, j \in Q$, where the functions $\mathbb{I}_0$, and $\mathbb{I}_a$ are indicator functions over the initial set of states, and states with label $\{a \}$ respectively.}
{
We then search for the piecewise components of the closure certificate via a counterexample-guided approach by collecting round $400$ points from the system.
To do so, we encode the conditions as a linear program, and set the s-procedure coefficients of $\tau_1 = 1$ for the conditions~\eqref{eq:tbar_cond_2_ineq} and the values of $\tau_2 = 1$, and $\tau_3 = 0$ for conditions~\eqref{eq:tbar_cond_3_ineq} to find a candidate closure certificate.
To speed up the search for counterexamples, we randomly sample points and check if the conditions fail to hold. If so we have found a counterexample.
If no such counterexample is found, we then formulate a query in dReal to search for a valid counterexample.
We repeat this process until no counterexamples are found.
The coefficients of the resulting closure certificate are displayed as a table in Appendix~\ref{ap:cc_2d_kur_safe}.
The time taken for this CEGIS loop to terminate is around $1$ hour and $50$ minutes on the reference machine.
We find the value of $\xi$ to be $1$.
}
\subsection{Two Room Temperature Model}
As a second case study, we consider our system $\Sys = (\Xx, \Xx_0, f)$ to be an interconnected two-room temperature model adapted from~\cite{anand_2021_compositional},
where $\Xx = [20, 34] \times [20, 34] \in \R^{2}$ indicate the temperature of the two rooms, $\Xx_0 = [21, 24] \times [21,24]$ indicate the initial set of states, and the transition function is defined as:
\[ f(x_1, x_2) = A \begin{bmatrix} x_1 \\ x_2 \end{bmatrix} + \mu T_h\begin{bmatrix} u(x_1) \\ u(x_2) \end{bmatrix} + \theta \begin{bmatrix} T_e \\ T_e \end{bmatrix},\]
where $x_i$ represents the temperature of room $i$, for all $i\in\{1, 2\}$, the matrix $A$ is 
\[
A := \begin{bmatrix}
1- 2\alpha - \theta - \mu u(x_1) & \alpha \\ 
 \alpha & 1- 2\alpha - \theta - \mu u(x_2) \\ 
\end{bmatrix},
\]
where constants $\alpha =0.004$, $\theta =0.01$, and $\mu =0.15$ represent the conduction factors, and $u(x)$ denotes the temperature controller, and is defined as $u(x_i) = 0.59 - 0.011x_i$. The value $T_h = 40$C denotes the heater temperature, and $T_e = 0$C represents the ambient temperature.
{Let the LTL formula to be verified be $ a_0\implies \mathsf{F} \mathsf{G} \neg a_1$.
This property requires that a system that starts from a state with atomic proposition $a_0$ does not visit the states with atomic proposition $a_1$ infinitely often.
We consider the atomic propositions $AP = \set{a_0, a_1}$, and the alphabet $\Sigma = \{ \{ \}, \{ a_0\}, \{ a_1\}, \{a_0,a_1 \} \}$.
In this setting, we require that if a state sequence of the system starts from $\Xx_0$ then it must visit the region $( [20,26]\times  \Xx) \cup (\Xx \times  [20,26] )$ finitely often. 
The complement of this specification is $a_0 \wedge \mathsf{G} \mathsf{F} a_1$ and the NBA $\Aa$ in Figure~\ref{fig:aut_case_study_live} denotes this complement.
Here, we mark the states $(x_1, x_2) \in \Xx_0$  with the atomic proposition $a_0$.
We mark a state $(x_1, x_2) \in \Xx$ with atomic proposition $a_1$, if $(x_1, x_2) \in ( [20,26] \cup \Xx) \times (\Xx \cup  [20,26]) $.
All other states are not marked with any atomic proposition.
Observe that a state $(x_1, x_2)$ may be marked with both atomic propositions $a_0$, and $a_1$, or neither. 
We define the labeling map as:
\[\Ll(x_1,x_2) = \begin{cases} \{a_0,a_1 \} & \text{ if } (x_1, x_2) \
\text{ is marked with both $a_0,a_1$} \\
\{a_0 \} & \text{ if } (x_1, x_2) 
\text{ is marked with only $a_0$} \\
\{a_1 \} & \text{ if } (x_1, x_2) 
\text{ is marked with only $a_1$} \\ 
\{ \} & \text{ otherwise, } 
\end{cases}\]
We consider the template of the piecewise components of the closure certificate to be specified as:
\begin{multline} 
\Tt_{i,j}\big( (x_1, x_2), (y_1, y_2) \big) = c_{0,i, j} + c_{1,i, j}   x_1  + c_{2, i, j}   x_2  +  c_{3, i, j}   y_1  \\ + c_{4,i, j}   y_2  + c_{5,i, j}   \max(x_1, x_2) + c_{6, i, j}   \max(y_1, y_2) \\
+ c_{7,i, j}   x_1^2 + c_{8,i,j} x_2^2
+ c_{9,i, j}   y_1^2 + c_{10,i, j}   y_2^2 ,
\end{multline}
for all states $(x_1, x_2)$,  and $(y_1, y_2) \in \Xx $ and all states $i,j$ of the NBA $\Aa$ in Figure~\ref{fig:aut_case_study_live}.
We then search for the piecewise components of the closure certificate using a CEGIS approach.
To speed up this, we first solve the linear program with around $100$ points, where we set the values of $\tau_1 = 1$, $\tau_2 = 0.4$, and  $\tau_3 = 0.1 $. 
We then search for counterexamples by first randomly sampling points, after which we use z3 to find counterexamples.
The resulting coefficients are described in a table in Appendix~\ref{ap:Case_study_two_room_cc}.
The time taken to find the closure certificate is around $1.5$ hours on the reference machine.
Finally, we find the value of $\xi$ to be $0.5$ in this example.
}

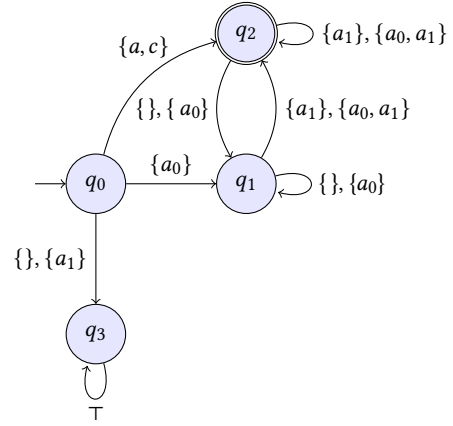
\begin{figure}[t]
    \centering
    \begin{tikzpicture}[node distance =2cm]
    \node[initial, state, draw, initial text =,fill=blue!10!white] (0) at (0,0) {$q_0$};
    \node[,state, fill=blue!10!white,] (3) at (0,-2) {$q_3$};
    \node[ ,state, fill=blue!10!white,] (1) at (2,0) {$q_1$};
    \node[accepting, state, fill = blue!10!white] (2) at (2, 2) {$q_2$};
    \path[->]
    (0) edge[] node[above]{$ \{ a_0\}$} (1)
    (0) edge[bend left] node[above=0.2cm]{ $\{ a,c\}$ } (2)
    (0) edge[] node[left]{$ \{ \}, \{a_1 \}$} (3)
    (3) edge[loop below] node[below]{$\top$} (3)
    (1) edge[bend right] node[right]{$\{a_1\}, \{ a_0,a_1 \}$} (2)
    (1) edge[loop right] node[right]{\{ \}, \{$a_0$\} } (2)
    (2) edge[loop right] node[right]{$\{a_1\}, \{ a_0,a_1 \}$} (2)
     (2) edge[bend right] node[left]{\{ \}, \{ $a_0$\} } (1);   
    \end{tikzpicture}
    \caption{A (nondeterministic) B\"uchi automaton $\Aa$ for the two-room temperature case study from Section~\ref{sec:case_studies}. 
    The automata represents the LTL formula $a_0 \wedge \mathsf{G} \mathsf{F} a_1$.
    Here $\top$ indicates any letter in the alphabet.}
    \label{fig:aut_case_study_live}
\end{figure}

\section{Conclusion}
We proposed a notion of so-called closure certificates that act as a function analog of transition invariants.
Our notion of closure certificates provide an abstraction-free approach to verify dynamical systems against $\omega$-regular properties. 
Our approach of using closure certificates to verify $\omega$-regular properties subsume existing approaches that use barrier certificate to verify $\omega$-regular properties.
As future work, we plan to investigate how one may use approaches such as $k$-induction to allow for a larger class of functions to act as closure certificates.
We also plan on investigating data driven approaches to find these closure certificates as well as investigate their use in synthesizing controllers.

\section{acknowledgements}
The authors thank Mateo Perez and Sriram Sankaranarayanan for valuable discussions as well as the anonymous reviewers for their constructive comments. 
This work was supported by NSF CAREER awards CCF-2146563, and CNS-2145184, and grants  ECCS-2015403, CNS-2039062, and CNS-2111688. 
 \bibliographystyle{ACM-Reference-Format}
\bibliography{bibliography.bib}


\begin{thebibliography}{37}

%%% the Web address.

\ifx \showCODEN    \undefined \def \showCODEN     #1{\unskip}     \fi
\ifx \showDOI      \undefined \def \showDOI       #1{#1}\fi
\ifx \showISBNx    \undefined \def \showISBNx     #1{\unskip}     \fi
\ifx \showISBNxiii \undefined \def \showISBNxiii  #1{\unskip}     \fi
\ifx \showISSN     \undefined \def \showISSN      #1{\unskip}     \fi
\ifx \showLCCN     \undefined \def \showLCCN      #1{\unskip}     \fi
\ifx \shownote     \undefined \def \shownote      #1{#1}          \fi
\ifx \showarticletitle \undefined \def \showarticletitle #1{#1}   \fi
\ifx \showURL      \undefined \def \showURL       {\relax}        \fi
% invisible to TeX
\providecommand\bibfield[2]{#2}
\providecommand\bibinfo[2]{#2}
\providecommand\natexlab[1]{#1}
\providecommand\showeprint[2][]{arXiv:#2}

\bibitem[Anand et~al\mbox{.}(2021)]%
        {anand_2021_compositional}
\bibfield{author}{\bibinfo{person}{Mahathi Anand}, \bibinfo{person}{Abolfazl Lavaei}, {and} \bibinfo{person}{Majid Zamani}.} \bibinfo{year}{2021}\natexlab{}.
\newblock \showarticletitle{Compositional Synthesis of Control Barrier Certificates for Networks of Stochastic Systems against $omega $-Regular Specifications}.
\newblock \bibinfo{journal}{\emph{arXiv preprint arXiv:2103.02226}} (\bibinfo{year}{2021}).
\newblock


\bibitem[Anand et~al\mbox{.}(2022)]%
        {anand_2022_small}
\bibfield{author}{\bibinfo{person}{Mahathi Anand}, \bibinfo{person}{Abolfazl Lavaei}, {and} \bibinfo{person}{Majid Zamani}.} \bibinfo{year}{2022}\natexlab{}.
\newblock \showarticletitle{From small-gain theory to compositional construction of barrier certificates for large-scale stochastic systems}.
\newblock \bibinfo{journal}{\emph{IEEE Trans. Automat. Control}} \bibinfo{volume}{67}, \bibinfo{number}{10} (\bibinfo{year}{2022}).
\newblock


\bibitem[Baier and Katoen(2008)]%
        {baier_2008_principles}
\bibfield{author}{\bibinfo{person}{Christel Baier} {and} \bibinfo{person}{Joost-Pieter Katoen}.} \bibinfo{year}{2008}\natexlab{}.
\newblock \bibinfo{booktitle}{\emph{Principles of model checking}}.
\newblock \bibinfo{publisher}{MIT press}.
\newblock


\bibitem[Barrett and Tinelli(2018)]%
        {barrett_2018_satisfiability}
\bibfield{author}{\bibinfo{person}{Clark Barrett} {and} \bibinfo{person}{Cesare Tinelli}.} \bibinfo{year}{2018}\natexlab{}.
\newblock \bibinfo{booktitle}{\emph{Satisfiability modulo theories}}.
\newblock \bibinfo{publisher}{Springer}.
\newblock


\bibitem[Collins(1975)]%
        {collins_1975_quantifier}
\bibfield{author}{\bibinfo{person}{George~E Collins}.} \bibinfo{year}{1975}\natexlab{}.
\newblock \showarticletitle{Quantifier elimination for real closed fields by cylindrical algebraic decomposition}. In \bibinfo{booktitle}{\emph{Automata Theory and Formal Languages: 2nd GI Conference Kaiserslautern, May 20--23, 1975}}. \bibinfo{publisher}{Springer}, \bibinfo{pages}{134--183}.
\newblock


\bibitem[Cook(2009)]%
        {cook_2009_priciples}
\bibfield{author}{\bibinfo{person}{Byron Cook}.} \bibinfo{year}{2009}\natexlab{}.
\newblock \showarticletitle{Priciples of program termination}.
\newblock \bibinfo{journal}{\emph{Engineering Methods and Tools for Software Safety and Security}} \bibinfo{volume}{22}, \bibinfo{number}{161} (\bibinfo{year}{2009}), \bibinfo{pages}{125}.
\newblock


\bibitem[Dutertre and Moura(2006)]%
        {dutertre_2006_fast}
\bibfield{author}{\bibinfo{person}{Bruno Dutertre} {and} \bibinfo{person}{Leonardo~de Moura}.} \bibinfo{year}{2006}\natexlab{}.
\newblock \showarticletitle{A fast linear-arithmetic solver for DPLL (T)}. In \bibinfo{booktitle}{\emph{International Conference on Computer Aided Verification}}. \bibinfo{publisher}{Springer}, \bibinfo{pages}{81--94}.
\newblock


\bibitem[Filiot et~al\mbox{.}(2009)]%
        {filiot_2009_antichain}
\bibfield{author}{\bibinfo{person}{Emmanuel Filiot}, \bibinfo{person}{Naiyong Jin}, {and} \bibinfo{person}{Jean-Fran{\c{c}}ois Raskin}.} \bibinfo{year}{2009}\natexlab{}.
\newblock \showarticletitle{An antichain algorithm for LTL realizability}. In \bibinfo{booktitle}{\emph{International Conference on Computer Aided Verification}}. \bibinfo{publisher}{Springer}, \bibinfo{pages}{263--277}.
\newblock


\bibitem[Gao et~al\mbox{.}(2012)]%
        {gao_complete_2012}
\bibfield{author}{\bibinfo{person}{Sicun Gao}, \bibinfo{person}{Jeremy Avigad}, {and} \bibinfo{person}{Edmund~M. Clarke}.} \bibinfo{year}{2012}\natexlab{}.
\newblock \showarticletitle{$\delta$-{complete} {decision} {procedures} for {satisfiability} over the {reals}}. In \bibinfo{booktitle}{\emph{Automated {Reasoning}}} \emph{(\bibinfo{series}{Lecture {Notes} in {Computer} {Science}})}. \bibinfo{publisher}{Springer}, \bibinfo{pages}{286--300}.
\newblock


\bibitem[Gao et~al\mbox{.}(2013)]%
        {gao_dreal_2013}
\bibfield{author}{\bibinfo{person}{Sicun Gao}, \bibinfo{person}{Soonho Kong}, {and} \bibinfo{person}{Edmund~M. Clarke}.} \bibinfo{year}{2013}\natexlab{}.
\newblock \showarticletitle{{dReal}: {An} {SMT} {solver} for {nonlinear} {theories} over the {reals}}. In \bibinfo{booktitle}{\emph{Automated {Deduction} – {CADE}-24}} \emph{(\bibinfo{series}{Lecture {Notes} in {Computer} {Science}})}. \bibinfo{publisher}{Springer}, \bibinfo{pages}{208--214}.
\newblock


\bibitem[Guo et~al\mbox{.}(2021)]%
        {GUO2021106804}
\bibfield{author}{\bibinfo{person}{Yufeng Guo}, \bibinfo{person}{Dongrui Zhang}, \bibinfo{person}{Zhuchun Li}, \bibinfo{person}{Qi Wang}, {and} \bibinfo{person}{Daren Yu}.} \bibinfo{year}{2021}\natexlab{}.
\newblock \showarticletitle{Overviews on the applications of the Kuramoto model in modern power system analysis}.
\newblock \bibinfo{journal}{\emph{International Journal of Electrical Power \& Energy Systems}}  \bibinfo{volume}{129} (\bibinfo{year}{2021}), \bibinfo{pages}{106804}.
\newblock


\bibitem[{Gurobi Optimization, LLC}(2021)]%
        {gurobi}
\bibfield{author}{\bibinfo{person}{{Gurobi Optimization, LLC}}.} \bibinfo{year}{2021}\natexlab{}.
\newblock \bibinfo{title}{{Gurobi Optimizer Reference Manual}}.
\newblock
\newblock
\urldef\tempurl%
\url{https://www.gurobi.com}
\showURL{%
\tempurl}


\bibitem[He et~al\mbox{.}(2020)]%
        {he_2020}
\bibfield{author}{\bibinfo{person}{Binghan He}, \bibinfo{person}{Jaemin Lee}, \bibinfo{person}{Ufuk Topcu}, {and} \bibinfo{person}{Luis Sentis}.} \bibinfo{year}{2020}\natexlab{}.
\newblock \showarticletitle{BP-RRT: Barrier pair synthesis for temporal logic motion planning}. In \bibinfo{booktitle}{\emph{2020 59th IEEE Conference on Decision and Control (CDC)}}. \bibinfo{publisher}{IEEE}, \bibinfo{pages}{1404--1409}.
\newblock


\bibitem[Henzinger et~al\mbox{.}(1997)]%
        {henzinger_1997_hytech}
\bibfield{author}{\bibinfo{person}{Thomas~A Henzinger}, \bibinfo{person}{Pei-Hsin Ho}, {and} \bibinfo{person}{Howard Wong-Toi}.} \bibinfo{year}{1997}\natexlab{}.
\newblock \showarticletitle{HyTech: A model checker for hybrid systems}. In \bibinfo{booktitle}{\emph{Computer Aided Verification: 9th International Conference, CAV'97 Haifa, Israel, June 22--25, 1997 Proceedings 9}}. Springer, \bibinfo{pages}{460--463}.
\newblock


\bibitem[Jagtap et~al\mbox{.}(2018)]%
        {jagtap_2018_temporal}
\bibfield{author}{\bibinfo{person}{Pushpak Jagtap}, \bibinfo{person}{Sadegh Soudjani}, {and} \bibinfo{person}{Majid Zamani}.} \bibinfo{year}{2018}\natexlab{}.
\newblock \showarticletitle{Temporal logic verification of stochastic systems using barrier certificates}. In \bibinfo{booktitle}{\emph{ATVA}}. \bibinfo{pages}{177--193}.
\newblock


\bibitem[Jagtap et~al\mbox{.}(2020)]%
        {jagtap_2020_formal}
\bibfield{author}{\bibinfo{person}{Pushpak Jagtap}, \bibinfo{person}{Sadegh Soudjani}, {and} \bibinfo{person}{Majid Zamani}.} \bibinfo{year}{2020}\natexlab{}.
\newblock \showarticletitle{Formal synthesis of stochastic systems via control barrier certificates}.
\newblock \bibinfo{journal}{\emph{IEEE Trans. Automat. Control}} \bibinfo{volume}{66}, \bibinfo{number}{7} (\bibinfo{year}{2020}), \bibinfo{pages}{3097--3110}.
\newblock


\bibitem[Jovanovi{\'c} and De~Moura(2013)]%
        {jovanovic_2013_solving}
\bibfield{author}{\bibinfo{person}{Dejan Jovanovi{\'c}} {and} \bibinfo{person}{Leonardo De~Moura}.} \bibinfo{year}{2013}\natexlab{}.
\newblock \showarticletitle{Solving non-linear arithmetic}.
\newblock \bibinfo{journal}{\emph{ACM Communications in Computer Algebra}} \bibinfo{volume}{46}, \bibinfo{number}{3/4} (\bibinfo{year}{2013}), \bibinfo{pages}{104--105}.
\newblock


\bibitem[Khaled and Zamani(2021)]%
        {khaled_2021_omegathreads}
\bibfield{author}{\bibinfo{person}{Mahmoud Khaled} {and} \bibinfo{person}{Majid Zamani}.} \bibinfo{year}{2021}\natexlab{}.
\newblock \showarticletitle{OmegaThreads: symbolic controller design for $\omega$-regular objectives}. In \bibinfo{booktitle}{\emph{Proceedings of the 24th International Conference on Hybrid Systems: Computation and Control}} (Nashville, Tennessee) \emph{(\bibinfo{series}{HSCC '21})}. \bibinfo{publisher}{Association for Computing Machinery}, \bibinfo{address}{New York, NY, USA}, Article \bibinfo{articleno}{25}, \bibinfo{numpages}{7}~pages.
\newblock
\showISBNx{9781450383394}
\urldef\tempurl%
\url{https://doi.org/10.1145/3447928.3456652}
\showDOI{\tempurl}


\bibitem[Kong et~al\mbox{.}(2018)]%
        {kong_2018_delta}
\bibfield{author}{\bibinfo{person}{Soonho Kong}, \bibinfo{person}{Armando Solar-Lezama}, {and} \bibinfo{person}{Sicun Gao}.} \bibinfo{year}{2018}\natexlab{}.
\newblock \showarticletitle{Delta-decision procedures for exists-forall problems over the reals}. In \bibinfo{booktitle}{\emph{Computer Aided Verification: 30th International Conference, CAV 2018, Held as Part of the Federated Logic Conference, FloC 2018, Oxford, UK, July 14-17, 2018, Proceedings, Part II 30}}. \bibinfo{publisher}{Springer}, \bibinfo{pages}{219--235}.
\newblock


\bibitem[Kozen(2006)]%
        {kozen_2006_theory}
\bibfield{author}{\bibinfo{person}{Dexter~C Kozen}.} \bibinfo{year}{2006}\natexlab{}.
\newblock \bibinfo{booktitle}{\emph{Theory of computation}}. Vol.~\bibinfo{volume}{121}.
\newblock \bibinfo{publisher}{Springer}.
\newblock


\bibitem[Lahijanian et~al\mbox{.}(2011)]%
        {lahijanian_2011_temporal}
\bibfield{author}{\bibinfo{person}{Morteza Lahijanian}, \bibinfo{person}{Sean~B Andersson}, {and} \bibinfo{person}{Calin Belta}.} \bibinfo{year}{2011}\natexlab{}.
\newblock \showarticletitle{Temporal logic motion planning and control with probabilistic satisfaction guarantees}.
\newblock \bibinfo{journal}{\emph{IEEE Transactions on Robotics}} \bibinfo{volume}{28}, \bibinfo{number}{2} (\bibinfo{year}{2011}), \bibinfo{pages}{396--409}.
\newblock


\bibitem[Moura and Bj{\o}rner(2008)]%
        {moura_2008_z3}
\bibfield{author}{\bibinfo{person}{Leonardo~de Moura} {and} \bibinfo{person}{Nikolaj Bj{\o}rner}.} \bibinfo{year}{2008}\natexlab{}.
\newblock \showarticletitle{Z3: An efficient SMT solver}. In \bibinfo{booktitle}{\emph{International conference on Tools and Algorithms for the Construction and Analysis of Systems}}. \bibinfo{pages}{337--340}.
\newblock


\bibitem[Parrilo(2003)]%
        {Parrilo_2003}
\bibfield{author}{\bibinfo{person}{Pablo~A. Parrilo}.} \bibinfo{year}{2003}\natexlab{}.
\newblock \showarticletitle{Semidefinite programming relaxations for semialgebraic problems}.
\newblock \bibinfo{journal}{\emph{Mathematical Programming}}  \bibinfo{volume}{96} (\bibinfo{year}{2003}), \bibinfo{pages}{293--320}.
\newblock


\bibitem[Pnueli(1977)]%
        {pnueli_1977_temporal}
\bibfield{author}{\bibinfo{person}{Amir Pnueli}.} \bibinfo{year}{1977}\natexlab{}.
\newblock \showarticletitle{The temporal logic of programs}. In \bibinfo{booktitle}{\emph{18th Annual Symposium on Foundations of Computer Science}}. IEEE, \bibinfo{pages}{46--57}.
\newblock


\bibitem[Podelski and Rybalchenko(2004)]%
        {podelski_2004_transition}
\bibfield{author}{\bibinfo{person}{Andreas Podelski} {and} \bibinfo{person}{Andrey Rybalchenko}.} \bibinfo{year}{2004}\natexlab{}.
\newblock \showarticletitle{Transition invariants}. In \bibinfo{booktitle}{\emph{Proceedings of the 19th Annual IEEE Symposium on Logic in Computer Science, 2004.}} IEEE, \bibinfo{pages}{32--41}.
\newblock


\bibitem[Podelski and Wagner(2006)]%
        {podelski_2006_model}
\bibfield{author}{\bibinfo{person}{Andreas Podelski} {and} \bibinfo{person}{Silke Wagner}.} \bibinfo{year}{2006}\natexlab{}.
\newblock \showarticletitle{Model checking of hybrid systems: From reachability towards stability}. In \bibinfo{booktitle}{\emph{International Workshop on Hybrid Systems: Computation and Control}}. Springer, \bibinfo{pages}{507--521}.
\newblock


\bibitem[Prajna and Jadbabaie(2004)]%
        {prajna_2004_safety}
\bibfield{author}{\bibinfo{person}{Stephen Prajna} {and} \bibinfo{person}{Ali Jadbabaie}.} \bibinfo{year}{2004}\natexlab{}.
\newblock \showarticletitle{Safety verification of hybrid systems using barrier certificates}. In \bibinfo{booktitle}{\emph{International Workshop on Hybrid Systems: Computation and Control}}. \bibinfo{publisher}{Springer}, \bibinfo{pages}{477--492}.
\newblock


\bibitem[Prajna et~al\mbox{.}(2002)]%
        {prajna_2002_introducing}
\bibfield{author}{\bibinfo{person}{Stephen Prajna}, \bibinfo{person}{Antonis Papachristodoulou}, {and} \bibinfo{person}{Pablo~A Parrilo}.} \bibinfo{year}{2002}\natexlab{}.
\newblock \showarticletitle{Introducing SOSTOOLS: A general purpose sum of squares programming solver}. In \bibinfo{booktitle}{\emph{Proceedings of the 41st IEEE Conference on Decision and Control, 2002.}} \bibinfo{publisher}{IEEE}, \bibinfo{pages}{741--746}.
\newblock


\bibitem[Prajna and Rantzer(2007)]%
        {prajna_2007_convex}
\bibfield{author}{\bibinfo{person}{Stephen Prajna} {and} \bibinfo{person}{Anders Rantzer}.} \bibinfo{year}{2007}\natexlab{}.
\newblock \showarticletitle{Convex programs for temporal verification of nonlinear dynamical systems}.
\newblock \bibinfo{journal}{\emph{SIAM Journal on Control and Optimization}} (\bibinfo{year}{2007}), \bibinfo{pages}{999--1021}.
\newblock


\bibitem[Rungger and Zamani(2016)]%
        {rungger_2016_scots}
\bibfield{author}{\bibinfo{person}{Matthias Rungger} {and} \bibinfo{person}{Majid Zamani}.} \bibinfo{year}{2016}\natexlab{}.
\newblock \showarticletitle{SCOTS: A tool for the synthesis of symbolic controllers}. In \bibinfo{booktitle}{\emph{Proceedings of the 19th international conference on hybrid systems: Computation and control}}. \bibinfo{pages}{99--104}.
\newblock


\bibitem[Safra(1988)]%
        {safra_1988_complexity}
\bibfield{author}{\bibinfo{person}{Shmuel Safra}.} \bibinfo{year}{1988}\natexlab{}.
\newblock \showarticletitle{On the complexity of $\omega$-automata}. In \bibinfo{booktitle}{\emph{Proc. 29th IEEE Symp. Found. of Comp. Sci}}. \bibinfo{publisher}{IEEE}, \bibinfo{pages}{319--327}.
\newblock


\bibitem[Sankaranarayanan and Tiwari(2011)]%
        {sankaranarayanan_2011_relational}
\bibfield{author}{\bibinfo{person}{Sriram Sankaranarayanan} {and} \bibinfo{person}{Ashish Tiwari}.} \bibinfo{year}{2011}\natexlab{}.
\newblock \showarticletitle{Relational abstractions for continuous and hybrid systems}. In \bibinfo{booktitle}{\emph{International Conference on Computer Aided Verification}}. \bibinfo{publisher}{Springer}, \bibinfo{pages}{686--702}.
\newblock


\bibitem[Schewe and Finkbeiner(2007)]%
        {schewe_2007_bounded}
\bibfield{author}{\bibinfo{person}{Sven Schewe} {and} \bibinfo{person}{Bernd Finkbeiner}.} \bibinfo{year}{2007}\natexlab{}.
\newblock \showarticletitle{Bounded synthesis}. In \bibinfo{booktitle}{\emph{International symposium on automated technology for verification and analysis}}. \bibinfo{publisher}{Springer}, \bibinfo{pages}{474--488}.
\newblock


\bibitem[Solar-Lezama(2008)]%
        {Lezama_2008_thesis}
\bibfield{author}{\bibinfo{person}{Armando Solar-Lezama}.} \bibinfo{year}{2008}\natexlab{}.
\newblock \bibinfo{booktitle}{\emph{Program Synthesis by Sketching}}.
\newblock


\bibitem[Tabuada(2009)]%
        {tabuada_2009_verification}
\bibfield{author}{\bibinfo{person}{Paulo Tabuada}.} \bibinfo{year}{2009}\natexlab{}.
\newblock \bibinfo{booktitle}{\emph{Verification and Control of Hybrid Systems: A Symbolic Approach}}.
\newblock \bibinfo{publisher}{Springer Science \& Business Media}.
\newblock


\bibitem[Vardi(2005)]%
        {vardi_2005_automata}
\bibfield{author}{\bibinfo{person}{Moshe~Y Vardi}.} \bibinfo{year}{2005}\natexlab{}.
\newblock \showarticletitle{An automata-theoretic approach to linear temporal logic}.
\newblock \bibinfo{journal}{\emph{Logics for concurrency: structure versus automata}} (\bibinfo{year}{2005}), \bibinfo{pages}{238--266}.
\newblock


\bibitem[Wongpiromsarn et~al\mbox{.}(2015)]%
        {wongpiromsarn_2015_automata}
\bibfield{author}{\bibinfo{person}{Tichakorn Wongpiromsarn}, \bibinfo{person}{Ufuk Topcu}, {and} \bibinfo{person}{Andrew Lamperski}.} \bibinfo{year}{2015}\natexlab{}.
\newblock \showarticletitle{Automata theory meets barrier certificates: Temporal logic verification of nonlinear systems}.
\newblock \bibinfo{journal}{\emph{IEEE Trans. Automat. Control}} \bibinfo{volume}{61}, \bibinfo{number}{11} (\bibinfo{year}{2015}), \bibinfo{pages}{3344--3355}.
\newblock


\end{thebibliography}
\newpage
\appendix

\section{Continuous space example for Theorem~\ref{thm:Bar_Tinv}}
\label{ap:thm_cont_BarTinv}
\begin{figure}[h]
  \begin{center}
  \begin{tikzpicture}[node distance =1.5cm]
    \node[state, draw, initial text =,fill=green!10!white] (0) at (0,0) {$0$};
    \node[state, draw, initial text =,fill=blue!10!white] (1) at (0,1.5) {$1$};

    \node[,state, fill=red!10!white,] (2) at (1.5,0) {$\frac{1}{32}$};
    \node[,state, fill=blue!10!white,] (3) at (1.5,1.5) {$\frac{33}{32}$};

    \node[,state, fill=green!10!white,] (4) at (3,0) {$\frac{1}{16}$};
    \node[,state, fill=blue!10!white,] (5) at (3,1.5) {$\frac{17}{16}$};

    \node[,state, fill=red!10!white,] (6) at (4.5,0) {$\frac{1}{8}$};
    \node[,state, fill=blue!10!white,] (7) at (4.5,1.5) {$\frac{9}{8}$};

    \node[,state, fill=green!10!white,] (8) at (6,0) {$\frac{1}{4}$};
    \node[,state, fill=blue!10!white,] (9) at (6,1.5) {$\frac{5}{4}$};

    \node[,state, fill=red!10!white,] (10) at (7.5,0) {$\frac{1}{2}$};
    \node[,state, fill=blue!10!white,] (11) at (7.5,1.5) {$\frac{3}{2}$};

    \node[] (12) at (0,3) {};
    \node[] (13) at (1.5,3) {};
    \node[] (14) at (3,3) {};
    \node[] (15) at (4.5,3) {};
    \node[] (16) at (6,3) {};
    \node[] (17) at (7.5,3) {};

    \node[] (18) at (0,-1.5) {};
    \node[] (19) at (3,-1.5) {};
    \node[] (20) at (6,-1.5) {};

    \path[->]
    (0) edge node{} (1)
    (2) edge node{} (3)
    (4) edge node{} (5)
    (6) edge node{} (7)
    (8) edge node{} (9)
    (10) edge node{} (11)
    (1) edge node{} (12)
    (3) edge node{} (13)
    (5) edge node{} (14)
    (7) edge node{} (15)
    (9) edge node{} (16)
    (11) edge node{} (17)
    (18) edge node{} (0)
    (19) edge node{} (4)
    (20) edge node{} (8)
    ;
    \end{tikzpicture}
  \end{center}
  \caption{A continuous-space system which has no barrier certificate of degree $2$ which indicates safety. The initial states are denoted in green, while the unsafe states are denoted in red. The dynamics are given by $f(x) = \set{ x +1}$ for all states $x \in \Xx$, where $\Xx = \R$ denotes the state space.}
\end{figure}
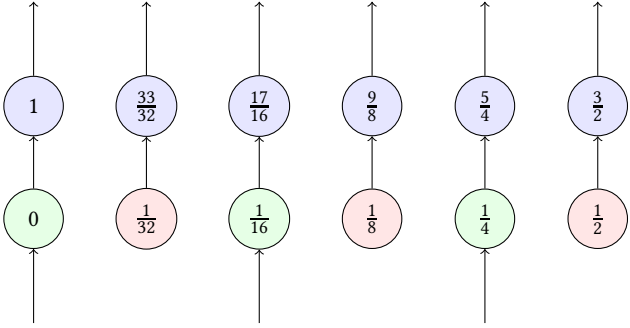
Observe that the above example requires any barrier certificate to be of a degree greater than $2$ by the intermediate value theorem.
Consider the function $\Tt(x,y) = y -x - 1$. This function satisfies condition~\eqref{eq:tbar_cond_1_safe}
as $\Tt(x, x+1) \geq 0$ for any $x \in \Xx$.
Second, if we have $\Tt(x+1, y) \geq 0$, \textit{i.e.}, $y - (x+1) -1 \geq 0$, then we can conclude that $(\Tt(x, y) \geq 0$ and so condition~\eqref{eq:tbar_cond_2_safe} also holds.
Finally for any initial state $x_0$, and unsafe state $x_u$, we have $\Tt(x_0, x_u) < 0$, thus ensuring condition~\eqref{eq:tbar_cond_3_safe} holds.
\section{Proof of Lemma~\ref{lem:triplet_twice}}
\label{ap:proof_lem_twice}
\begin{proof}

First let us unroll the simple cycles of the automaton $\Aa$ that start from an accepting state to construct an automaton $\Aa' = (\Sigma, Q', Q_0, \delta', Q'_{Acc})$ such that $L(\Aa') = L(\Aa)$.  
The states $Q'$ of $\Aa'$ contain all the states $Q$ of $\Aa$ as well as one new state $i'$ for every state $i \in Q$ that can be reached from an accepting state.
Determining which states are reachable from an accepting state can be done in polynomial time~\cite{kozen_2006_theory} in the number of states of the automaton.
Intuitively, these new states are meant to determine the state runs of the system which have reached an accepting state and continue onward.
The set of initial states is  the same as $\Aa$, while the set of accepting states consist of the states $ {\ell'}$, \textit{i.e.}, those accepting states that can be reached from an accepting state.
The transitions $\delta'$ are the same as the transitions for $\Aa$, except for those the transitions that leave an accepting state $ \ell \in {Acc}$ and additional transitions added.
The  transitions added are due to two reasons :
\begin{enumerate}
    \item For every pair of states $i, j \in Q$ that are reachable from an accepting state, and for any $\sigma \in \Sigma$ such that $j \in \delta(i,\sigma) $, we add a transition $j' \in \delta(i', \sigma)$.
    This preserves all the transitions for the newly added states so they behave the same as before. 
    \item For every state $i \in Q$, and every accepting state $\ell \in Acc$, and any letter $\sigma \in \Sigma$ such that $j \in \delta(\ell,\sigma) $, we add a transition $j' \in \delta(\ell, \sigma)$. For every edge that leaves the accepting state, we add an edge to move to the states of the form $q' \in Q' \setminus Q$. Observe that there are no transition from the states in $Q' \setminus Q$ to the ones in $Q$.
\end{enumerate}
Unrolling the automaton once, allows us to now consider the state triplets of the form $(q, q', q'')$, where either $q'$ and $q''$ are in $Q' \setminus Q$ or $q', q' \in Q$ and $q'' \in Q' \setminus Q$.

{
Unfortunately performing the same construction again does not add any new pair of letters to consider.
Let us repeat the above construction on NBA $\Aa'$ to construct the NBA $\Aa'' = (\Sigma, Q'', Q_0, \delta, Q''_f)$, and consider a state triplet $(q, q', q'')$ in $\Aa''$.
First, if all the states $q, q'$, and $q''$ are in the set $Q'$, then they were already a state triplet in the original NBA.
We now consider the different cases where a new state triplet $(q, q', q'')$ may be considered.
\begin{enumerate}
    \item If all the states in the triplet are in $Q''$ or in $Q'$, then there are no new edge pairs to consider (as there exist triplets $(q_1, q'_1, q''_1)$ in $\Aa'$ that have the same labels and have already been considered).
\item The state $q'' \in Q'' \setminus Q$ is newly added while  the states $q$ and $q'$ are in $Q'$. In this case there must be analogous state triplet $(q_1, q'_1, q''_1)$ such that $q''_1 \in Q' \setminus Q$, and the states $q_1$ and $q'_1$ were initially present in $Q$ which correspond to the same pair of letters. 
Thus, again there are no pairs to be considered.
\item The states $q'$ and $q'' $ are in the set $Q'' \setminus Q'$ (they are newly added states) but $q$ is not.
In this case there exists an analogous state triplet $(q_1, q'_1, q''_1)$, such that the state $q_1 \in Q$ and the states $q'_1$ and $q''_1$ are in the set $Q' \setminus Q$ with the same label.
Hence, this again does not add a new edge pair to be considered.
\end{enumerate} 
No other cases need to  be considered, as no state in $Q'' \setminus Q'$ can reach a state in $Q'$ by construction.
As no new pairs of letters are considered when unrolling more than once, we observe that we cannot find barrier certificates with a different initial or unsafe condition.
}
\end{proof}

\section{Proof of subsumption}
\label{ap:subsumption_proof}
\begin{theorem}[Subsuming the state triplet approach]
    Consider a System $\Sys = (\Xx, \Xx_0, f)$, a labeling map $\Ll: \Xx \to \Sigma$, and an NBA $\Aa' = (\Sigma, Q', Q_0, \delta', Q'_{Acc}) $ representing the complement of an $\omega$-regular specification, and let barrier certificates $\Bb_1, \ldots, \Bb_m$ be used to show that $TR(\Sys, \Ll) \cap L(\Aa') = \emptyset$ via the state triplet approach. 
    Then there exists a closure certificate $\Tt$ that also acts as a proof that  $TR(\Sys, \Ll) \cap L(\Aa') = \emptyset$.
\end{theorem}
\begin{proof}
    Fix $\xi \in \R_{ > 0}$ as a small positive value, then we define a closure certificate $\Tt: \Xx \times Q\times \Xx \times Q \to \R $ based on the following cases:
\begin{enumerate}
    \item The automaton states are either in $Q_l$ or $Q_r$.
    \item The first automaton state is a middle element of some triplet.
    \item The second automaton state is the middle element of some triplet.
\end{enumerate}
We define $\Tt(x,i,y,j)$ for each of the above cases as follows:
\begin{enumerate}
    \item
    We define $\Tt(x,i,y,j) = 0$, if $i$ and $j$ are in the same set, \textit{i.e.} both $i$ and $j$ are either in $Q_l$ or $Q_r$; or if the state $i \in Q_r$, and state $j \in Q_l$ and for every pair of states $x, y \in \Xx$.
    We define $\Tt(x,i,y,j) = -\xi$ if $i \in Q_l$, and $j \in Q_r$.
    Intuitively, this assumes that all states in $Q_l$ are reachable from any state in either $Q_r$ or $Q_l$, however no state in $Q_r$ is reachable from some state in $Q_l$.
    \item Let $i = q_m'$ be a middle element of some triplet, then we define $\Tt(x, i,y,j) $ depending on the value of $\Bb_m(x)$ and $\Bb_m(y)$.
    For all states $x \in \Xx$,  such that $\Bb_m(x) \leq 0$, we  define  $\Tt(x,i,y,j) = 0$ for all states $y \in \Xx$ such that $\Bb_m(y) \leq 0$ and $j \in Q_l \cup \set{i}$; and $\Tt(x,i,y,j) = -\xi$ for all states $j \in Q_r$ and $y \in \Xx$.
    Otherwise, we define $\Tt(x,i,y,j) = 0$ for all states $y \in \Xx$, and $j \in Q$.
    Intuitively, this separates those states that can be reached from $Q_l$, and those states that are in $Q_r$, while allowing the states in $Q_r$ to reach the states in $Q_l$.
    \item  Finally, let $j = q_m'$ be a middle element of some triplet, then we define $\Tt(x,i,y,j) = 0$ for all states $x \in \Xx$, $i \in Q_l$, and every state $y \in \Xx$ such that $\Bb_m(y) \leq 0$ and $\Tt(x,l, y, j) = -\xi$ otherwise.
    Similarly, we define $\Tt(x,i,y,j) = 0$ for all states $x \in \Xx$, $i\in Q_r$, and every state $y \in \Xx$.
    Intuitively, this enforces that starting from a state in $Q_l$,  one does not reach a state where the barrier value is greater than $0$, and the automaton state is the middle element of a triplet.
\end{enumerate}

We now show that $\Tt$ is a closure certificate.
First, we observe that condition~\eqref{eq:tbar_prod_cond_3} trivially holds as all states $qs \in Q_0$ are in the set $Q_l$ and all states $\ell \in Acc$ are in the set $Q_r$, therefore by definition we must have $\Tt(x,s, y, \ell) = -\xi$ for all states $x \in  \Xx_0$, and $y \in \Xx$.
Second, let us assume that condition~\eqref{eq:tbar_prod_cond_1} fails to hold, \textit{i.e.} there exists a state $i \in Q$, and $x \in \Xx$, $x' = f(x)$ and $i' \in \delta(i, \Ll(x))$, such that $\Tt(x,i, x', i') < 0 $.
For this to be true we must have $ \Tt(x, i, x', i') = -\xi$.
Clearly, $i \in Q_l$ and $i' \in Q_r$ cannot hold as the barrier certificates disallow reaching the states in $Q_r$ from $Q_l$ in one step of the transition without passing through the middle element of some triplet.
Thus either $i$ or $i'$ must be the middle element of some triplet.
If $i$ is the middle element of some triplet, then we must have $i' \in Q_r$, and $\Bb_m(x) \leq 0$.
As $\Bb_m(x) \leq 0$, we must have $\Bb_m(x') \leq 0$ as well and so $i'$ cannot be the last element of the triplet.
If $i'$ is another state in $Q_r$ that is to the right of every triplet, then it could not be reached from any state in $Q_l$ and hence cannot be reached from any state that is reachable from $i'$ and therefore $i$ cannot be a middle element of some triplet.
Lastly, if $i'$ is the middle element of some triplet, then it must be the case that $\Bb_m(y) > 0$, and following the last condition of the barrier no state in $Q_l$ should be able to reach $i'$.
Thus condition~\eqref{eq:tbar_prod_cond_2} must hold.
Finally,  assume that condition~\eqref{eq:tbar_prod_cond_3} fails to hold, then there exists some state $x \in \Xx$, $y \in  \Xx$, and states $i$, $j \in Q$, such that for $x' = f(x)$, and $i' \in \delta(i, \Ll(x))$, and we have 
$\Tt\big((x', i'), (y, j) \big) \geq 0$ and $\Tt\big((x, i), (y,j))  = -\xi$.
For this to be true, one of the following must hold:
\begin{enumerate}
\item the state $i$ is in $Q_l$, and the state $j$ is in $Q_r$.
We know that $\Tt((x',i'), (y, j)) \geq 0$, and so we may be able to reach state $j$ from $i'$. 
We know $i' \in \delta(i, \Ll(x))$, by definition, and so we must have $\Tt\big((x, i), (x', i') \big)\geq 0$, and as $i$ is in $Q_l$, it must follow that state $i'$ is also either in $q_l$ or the middle element of some triplet such that $\Bb_m(x') \leq 0$ for some $0 \leq m \leq k$.
From the previous conditions, it follows, that $\Tt\big((x', i'), (y, j) \big) = -\xi$.
\item The state $i$ is the middle element of some triplet, such that $\Bb_m(x) \leq 0$ and either $j$ is in $Q_r$ or  $\Bb_i(y) > 0$.
If $j$ is in $Q_r$, we observe that $\Tt(x', i', y, j) \geq 0$, and so we must have $\Bb_m(x') > 0$ which cannot hold or $i'$ is in $Q_r$.
However, $i' \in \delta(i, \Ll(x))$, and so similar to the previous case, this cannot be true.
\item Finally, the state $j$ is the middle element of some triplet.
If so, we must have $i \in Q_l$, and $\Bb_i(y) > 0$.
We observe that the we must have $\Tt(x, i, x', i') \geq 0$ as the first condition holds, and so $i' \in Q_l$.
If $\Tt(x',i', y, j) \geq 0$, then it must be the case that $\Bb_m(y) \leq 0$ by construction as $i' \in Q_l$.
This is a contradiction.
\end{enumerate}
\end{proof}

\section{Closure Certificates}
\label{ap:cl_cc}
We now state the values for the coefficients in our two case studies.
\subsection{Closure Certificate for the two-dimensional Kuramoto Oscillator and NBA in Figure~\ref{fig:aut_case_study_1_safe}}
\label{ap:cc_2d_kur_safe}

The coefficients for the closure certificate $\Tt \big((x_1, x_2, i), (y_1, y_2, j) \big)$ are specified in the table below:
\begin{center}
\begin{tabular}{ |c | c |c| c| c| c |c| c| c| }
\hline
$i$ & $j$ & $c_{0,i,j}$ &  $c_{1,i,j}$ &  $c_{2,i,j}$ &
 $c_{3,i,j}$ &  $c_{4,i,j}$ &  $c_{5,i,j}$ &  $c_{6,i,j}$ \\
\hline
$q_0$ & $q_0$ & $10$ &  $10$ &  $10$ &
 $10$ &  $10$ &  $10$ &  $10$  \\

$q_0$ & $q_1$ & $-0.5875
$ &  $10$ &  $1$0 &
 $0$ &  $0$ &  $-8.2456$ &  $-10$ \\

$q_1$ & $q_0$ & $-10$ &  $10$ &  $10$ &
 $10$ &  $10$ &  $-10$ &  $-10$ \\

$q_1$ & $q_1$ & $-1.4164$ &  $10$ &  $10$ &
 $10$ &  $10$ &  $1.7536$ &  $-0.2975$ \\

 \hline
\end{tabular}
\end{center}

\subsection{Closure Certificate for the Two room temperature example and NBA in Figure~\ref{fig:aut_case_study_live}}
\label{ap:Case_study_two_room_cc}
The coefficients for the closure certificate $\Tt \big((x_1, x_2, i), (y_1, y_2, j) \big) $ are given in the tables below:
\begin{center}
\begin{tabular}{ |c | c |c| c| c| c |c| c| }
\hline
$i$ & $j$ & $c_{0,i,j}$ &  $c_{1,i,j}$ &  $c_{2,i,j}$ &
 $c_{3,i,j}$ &  $c_{4,i,j}$ &  $c_{5,i,j}$  \\
\hline
$q_0$ & $q_0$ & $10$ &  $10$ &  $10$ &
 $10$ &  $10$ &  $10$   \\
 $q_0$ & $q_1$ & $10$ &  $10$ &  $10$ &
 $10$ &  $10$ &  $10$   \\
 $q_0$ & $q_2$ & $10$ &  $10$ &  $10$ &
 $10$ &  $10$ &  $10$   \\
 $q_0$ & $q_3$ & $10$ &  $10$ &  $10$ &
 $10$ &  $10$ &  $10$   \\

$q_1$ & $q_0$ & $-10$ &  $-10$ &  $-10$ &
 $-10$ &  $-10$ &  $-0.497$   \\
 $q_1$ & $q_1$ & $-10$ &  $-10$ &  $-10$ &
 $-10$ &  $-10$ &  $-0.497$   \\
$q_1$ & $q_2$ & $-10$ &  $-10$ &  $-10$ &
 $-10$ &  $-10$ &  $-0.497$   \\
$q_1$ & $q_3$ & $-10$ &  $-10$ &  $-10$ &
 $-10$ &  $-10$ &  $-0.497$   \\

$q_2$ & $q_0$ & $-10$ &  $-10$ &  $-10$ &
 $-10$ &  $-10$ &  $-1.47$   \\
$q_2$ & $q_1$ & $-10$ &  $-10$ &  $-10$ &
 $-10$ &  $-10$ &  $-1.47$   \\
$q_2$ & $q_2$ & $-10$ &  $-10$ &  $-10$ &
 $-10$ &  $-10$ &  $-1.47$   \\
$q_2$ & $q_3$ & $-10$ &  $-10$ &  $-10$ &
 $-10$ &  $-10$ &  $-1.47$   \\

 $q_3$ & $q_0$ & $-10$ &  $-10$ &  $-10$ &
 $-10$ &  $-10$ &  $-0.387$   \\
 $q_3$ & $q_1$ & $-10$ &  $-10$ &  $-10$ &
 $-10$ &  $-10$ &  $-0.387$   \\
 $q_3$ & $q_2$ & $-10$ &  $-10$ &  $-10$ &
 $-10$ &  $-10$ &  $-0.387$   \\
 $q_3$ & $q_3$ & $1$ &  $1$ &  $1$ &
 $1$ &  $1$ &  $1$   \\

 \hline
\end{tabular}
\end{center}

\begin{center}
\begin{tabular}{ |c | c |c| c| c| c |c| }
\hline
$i$ & $j$ & $c_{6,i,j}$ &  $c_{7,i,j}$ &  $c_{8,i,j}$ &
 $c_{9,i,j}$ &  $c_{10,i,j}$  \\
\hline
$q_0$ & $q_0$ & $10$ &  $10$ &  $10$ &
 $10$ &  $10$  \\
 $q_0$ & $q_1$ & $10$ &  $10$ &  $10$ &
 $10$ &  $10$   \\
 $q_0$ & $q_2$ & $10$ &  $10$ &  $10$ &
 $10$ &  $10$    \\
 $q_0$ & $q_3$ & $10$ &  $10$ &  $10$ &
 $10$ &  $10$   \\

$q_1$ & $q_0$ & $-10$ &  $0.179$ &  $-10$ &
 $0.169$ &  $-10$   \\
$q_1$ & $q_1$ & $-10$ &  $0.179$ &  $-10$ &
 $0.169$ &  $-10$   \\
$q_1$ & $q_2$ & $-10$ &  $0.179$ &  $-10$ &
 $0.169$ &  $-10$   \\
$q_1$ & $q_3$ & $-10$ &  $0.179$ &  $-10$ &
 $0.169$ &  $-10$   \\

$q_2$ & $q_0$ & $-10$ &  $0.194$ &  $-10$ &
 $0.191$ &  $-10$  \\
$q_2$ & $q_1$ & $-10$ &  $0.194$ &  $-10$ &
 $0.191$ &  $-10$   \\
$q_2$ & $q_2$ & $-10$ &  $0.194$ &  $-10$ &
 $0.191$ &  $-10$   \\
 $q_2$ & $q_3$ & $-10$ &  $0.194$ &  $-10$ &
 $0.191$ &  $-10$  \\
 
 $q_3$ & $q_0$ & $-10$ &  $0.177$ &  $-10$ &
 $0.163$ &  $-10$  \\
 $q_3$ & $q_1$ & $-10$ &  $0.177$ &  $-10$ &
 $0.163$ &  $-10$  \\
 $q_3$ & $q_2$ & $-10$ &  $0.177$ &  $-10$ &
 $0.163$ &  $-10$  \\
 $q_3$ & $q_3$ & $1$ &  $1$ &  $1$ &
 $1$ &  $1$   \\
 \hline
\end{tabular}
\end{center}

\newpage

\end{document}